\documentclass[12pt,a4paper]{article}
\usepackage{epsf, graphicx}
\usepackage{latexsym,amsfonts,amsbsy,amssymb}
\usepackage{amsmath,amsthm}
\usepackage{cite}
\usepackage{indentfirst}
\usepackage{bm}
\usepackage{cleveref}

\makeatletter

\@addtoreset{equation}{section} \makeatother
\newtheorem{theorem}{Theorem}[section]
\newtheorem{lemma}{Lemma}[section]
\newtheorem{corollary}{Corollary}[section]
\newtheorem{remark}{Remark}[section]
\newtheorem{definition}{Definition}[section]
\newtheorem{proposition}{Proposition}[section]
\newtheorem{example}{Example}[section]
\makeatletter \@addtoreset{equation}{section} \makeatother
\DeclareMathOperator{\wt}{wt}
\begin{document}
\title{Linear Codes from Simplicial Complex over $\mathbb{F}_{2^n}$}

\author{Hongwei Liu,\quad Zihao Yu}
\date{School of Mathematics and Statistics,
 Central China Normal University, Wuhan 430079, China}
\maketitle
\insert\footins{\small{\it Email addresses} :
 hwliu@mail.ccnu.edu.cn (Hongwei Liu); zihaoyu@mails.ccnu.edu.cn (Zihao Yu).}
{\centering\section*{Abstract}}
 \addcontentsline{toc}{section}{\protect Abstract}
 \setcounter{equation}{0}
In this article we mainly study linear codes over $\mathbb{F}_{2^n}$ and their binary subfield codes. We construct linear codes over $\mathbb{F}_{2^n}$ whose defining
sets are the certain subsets of $\mathbb{F}_{2^n}^m$ obtained from mathematical objects called simplicial complexes. We use a result in LFSR sequences to illustrate
the relation of the weights of codewords in two special codes obtained from simplical complexes and then determin the parameters of these codes. We construct five
infinite families of distance optimal codes and give sufficient conditions for these codes to be minimal.

\medskip
\noindent{\large\bf Keywords: }\medskip linear code, subfield code, simplicial complex, Griesmer code.

\noindent{\bf2010 Mathematics Subject Classification}: 94B05, 94B65.

\section{Introduction}
Let $\mathbb{F}_{q^n}$ be the finite field with $q^n$ elements, where $q$ is a power of prime. Given a linear code $C$ over $\mathbb{F}_{q^n}$, Ding and Heng
\cite{ding2019subfield} constructed a new linear code $C^{(q)}$ over $\mathbb{F}_q$, which is termed as the subfield code with respect to $C$. Chang and Hyun
\cite{chang2018linear} obtained a class of optimal binary linear codes and a class of minimal codes using Boolean functions associated with simplicial complexes. In
the later work, Hyun et al. \cite{hyun2020infinite} obtained some infinite families of optimal binary linear codes by choosing the specific defining set from
simplical complexes. More generally, Hyun et al. \cite{hyun2020optimal} used posets in place of simplicial complexes and obtained some optimal and minimal binary
linear codes which are not satisfying the Ashikhmin-Barg condition in \cite{ashikhmin1998minimal}. From then on, more works about linear codes via simplical complexes
are studied over $\mathbb{F}_{2^r}$, $r=1,2,3$. In \cite{zhu2021few}, Zhu and Wei constructed optimal and minimal quaternary  linear codes and gave their weight
distributions. Wu et al. \cite{wu2022quaternary} studied quaternary linear codes and their binary subfield codes, they presented two infinite families of optimal
linear codes. In \cite{sagar2022linear}, the authors used simplicial complexes to construct octanary linear codes and studied their corresponding binary subfield
codes. They obtained five infinite families of distance optimal linear codes and gave sufficient conditions of some of these codes to be minimal. Furthermore,
they proposed a few conjectures about the results of linear codes from simplical complexes over arbitrary $\mathbb{F}_{2^n}$.

Based on the above researches, in this article, we will study linear codes over arbitrary $\mathbb{F}_{2^n}$ whose defining set are obtained from simplical complexes
and their corresponding binary subfield codes. We will prove the conjectures proposed in \cite{sagar2022linear} by using the LFSR sequences and obtain some distance
optimal and minimal codes. Some general results in the previous works can be found in this article.

The rest of this article is arranged as follows. In Section 2, we recall some concepts in coding theory. We also introduce the definitions and properties of
simplicial complexes and the LFSR sequences. In Section 3, we give the relation of the weights of codewords in two codes obtained from simplical complexes. In Section
4, we determin the parameters of these codes and obtain some infinite families of distance optimal and minimal codes. In Section 5, we study the subfield codes with
respect to the codes in the previous sections. In Section 6, we draw a conclusion about our article.
\section{Preliminaries}
\subsection{The basic concepts in coding theory}
Throughout the article, for $m\in \mathbb{N}$, we denote the set $\{1,2,\ldots,m\}$ by $[m]$. Let $C$ be an $[n,k,d]$-linear code over $\mathbb{F}_q$.
Let $A_i$ be the number of codewords of weight $i$ in $C$. The sequence $(1,A_1,\ldots,A_n)$ is called the weight distribution of $C$. In addition,
if the cardinanity of the set $\{1\leq i\leq n:A_i\neq0\}$ is $t$, then $C$ is called a $t$-$weight$ linear code. The code $C$ is called distance optimal if
there exists no $[n,k,d+1]$-linear code.

Next we recall the well known Griesmer bound on linear codes.
\begin{lemma}{\rm(\cite{griesmer1960bound})}
    \textup{(Griesmer Bound)} If $C$ is an $[n,k,d]$-linear code over $\mathbb{F}_q$, then we have
    \begin{equation}\label{Griesmer Bound}
        \sum_{i=0}^{k-1}\left\lceil\frac{d}{q^i}\right\rceil \leq n,
    \end{equation}
    where $\left\lceil\cdot\right\rceil$ denotes the ceiling function.
\end{lemma}
An $[n,k,d]$-linear code is called a Griesmer code if the equality holds in \eqref{Griesmer Bound}. Note that every Griesmer code is distance optimal.

Let $D=\{d_1<d_2<\cdots<d_n\}$ be an ordered subset of $\mathbb{F}_q^m$. We can construct a linear code defined by
\begin{equation}
    C_D=\{(v\cdot d_1,v\cdot d_2,\cdots,v\cdot d_n):v\in \mathbb{F}_q^m\},
\end{equation}
where $x\cdot y=\sum_{i=1}^n x_iy_i$ for $x=(x_1,x_2,\cdots,x_m),y=(y_1,y_2,\cdots,y_m)$. The set $D$ is called the defining set of $C_D$.

For $x=(x_1,x_2,\cdots,x_m)\in\mathbb{F}_q^m$, the set $$\operatorname{Supp}(x)=\{i\in[m]:x_i\neq 0\}$$ is called the support of $x$. Note that the Hamming
weight of $x\in\mathbb{F}_q^m$ is $\wt(x)=|\operatorname{Supp}(x)|$.

For $x,y\in \mathbb{F}_q^m$, if $\operatorname{Supp}(y)\subseteq \operatorname{Supp}(x)$, we say $x$ covers $y$. Let $C$ be a code over $\mathbb{F}_q$.
A codeword $u_0$ in $C\backslash\{0\}$ is called minimal if $u\in C\backslash\{0\}$, $u_0$ covers $u$ implies $u=\lambda u_0$ for some $\lambda\in \mathbb{F}_q^*$.
A code $C$ over $\mathbb{F}_q$ is called a minimal code if every nonzero codeword of $C$ is a minimal codeword. There is a sufficient condition for a linear code to
be minimal.
\begin{lemma}\label{minimal} {\rm(\cite{ashikhmin1998minimal})}
    If the ritio of the minimum and the maximum Hamming weights of nonzero codewords in a linear code $C$ over $\mathbb{F}_q$ is larger than $\frac{q-1}{q}$, then
    $C$ is a minimal code.
\end{lemma}

Next we introduce subfield codes and their properties. Given an $[l,k]$ linear code over $\mathbb{F}_{q^n}$ generated by $G$, let $\{v_1,v_2,\cdots,v_n\}$ be a basis
of $\mathbb{F}_{q^n}$ over $\mathbb{F}_q$. For each entry $g_{ij}$ of $G$, $g_{ij}$ has the unique representation
$g_{ij}=g_{ij}^{(1)}v_1+g_{ij}^{(2)}v_2+\cdots+g_{ij}^{(n)}v_n$,
where $g_{ij}^{(k)}\in \mathbb{F}_q$ for $1\leq k\leq n$. Write
\begin{equation*}
    \begin{aligned}
        G_{ij}&=\left(\begin{matrix}
            g_{ij}^{(1)}\\
            g_{ij}^{(2)}\\
            \vdots\\
            g_{ij}^{(n)}
        \end{matrix}\right).
    \end{aligned}
\end{equation*}
Then the subfield code $C^{(q)}$ over $\mathbb{F}_q$ is the code generated by the matrix which is obtained by replacing each enry $g_{ij}$ of $G$ by $G_{ij}$. Ding 
and Heng in \cite{ding2019subfield} proved that the subfield code $C^{(q)}$ is independent of the choice of the basis of $\mathbb{F}_{q^n}$.

Now set $q=2$. Let $f(x)=x^n+a_{n-1}x^{n-1}+\cdots+a_1x+a_0\in \mathbb{F}_2[x]$ be an irreducible polynomial and $\omega$ a root of $f(x)$. Then
$\mathbb{F}_{2^n}=\mathbb{F}_2(\omega)$ and $\{1,\omega,\cdots,\omega^{n-1}\}$ is a basis of $\mathbb{F}_{2^n}$ over $\mathbb{F}_2$.
For $G\in M_{k\times l}(\mathbb{F}_{2^n})$, $G$ can be uniquely written as $G=G_0+\omega G_1+\cdots+\omega^{n-1}G_{n-1}$ where $G_i\in M_{k\times l}(\mathbb{F}_2)$.
For each entry $g_{ij}$ of $G$, let $g_{ij}=g_{ij}^{(0)}+\omega g_{ij}^{(1)}+\cdots+\omega^{n-1}g_{ij}^{(n-1)}$ where $g_{ij}^{(k)}\in G_k$ for $k=0,1,\cdots,n-1$.
Now we give the following result.
\begin{theorem}\label{subfield}
    Suppose $\mathbb{F}_{2^n}=\mathbb{F}_2(\omega)$, where $\omega$ is a root of an irreducible polynomial of degree $n$ over $\mathbb{F}_2$. Let $C$ be an $[l,k]$
    linear code over $\mathbb{F}_{2^n}$ with generator matrix $G=G_0+\omega G_1+\cdots+\omega^{n-1}G_{n-1}$, where $G_i\in M_{k\times l}(\mathbb{F}_2)$. Then the
    subfield code $C^{(2)}$ with respect to $C$ is generated by
    \begin{equation*}
        G^{(2)}=\left(\begin{matrix}
            G_0\\
            G_1\\
            \vdots\\
            G_{n-1}\\
        \end{matrix}\right).
    \end{equation*}
    In addition, if $C$ has the defining set $D=D_0+\omega D_1+\cdots+\omega^{n-1}D_{n-1}$ with $D_i\subseteq \mathbb{F}_2^m$ for $0\leq i\leq n-1$, then $C^{(2)}$ has
    defining set
    \begin{equation*}
        D^{(2)}=\{(d_0,d_1,\cdots,d_{n-1}):d_i\in D_i,0\leq i\leq n-1\}.
    \end{equation*}
\end{theorem}
\begin{proof}
    Let
    \begin{equation*}
        G_{ij}=\left(\begin{matrix}
            g_{ij}^{(0)}\\
            g_{ij}^{(1)}\\
            \vdots\\
            g_{ij}^{(n-1)}\\
        \end{matrix}\right).
    \end{equation*}
    Then the subfield code $C^{(2)}$ is generated by
    \begin{equation*}
        \left(\begin{matrix}
            G_{11}&G_{12}&\cdots&G_{1l}\\
            G_{21}&G_{22}&\cdots&G_{2l}\\
            \vdots&\vdots&\ddots&\vdots\\
            G_{k1}&G_{k2}&\cdots&G_{kl}\\
        \end{matrix}\right).
    \end{equation*}
    Reset the rows of the above matrix we can get $C^{(2)}$ is generated by
    \begin{equation*}
        G^{(2)}=\left(\begin{matrix}
            G_0\\
            G_1\\
            \vdots\\
            G_{n-1}\\
        \end{matrix}\right).
    \end{equation*}
    The rest part of this theorem is a straightforward result from the above conclusion.
\end{proof}

\subsection{Simplicial complex}
A subset $\Delta\subseteq \mathbb{F}_2^m$ is called a $simplicial~complex$ if $v\in \Delta$ and $v$ covers $u$ implies $u\in \Delta$.
An element $v_0$ in a simplicial complex $\Delta$ is called maximal if $v\in \Delta$, $v$ covers $v_0$ implies $v=v_0$.
It is easy to see that a simplicial complex $\Delta$ is uniquely determined by its maximal elements. For the simplest case, if a simplicial complex has only one
maximal element with support  $L\subseteq [m]$, we say the simplicial complex is generated by $L$, denoted by $\Delta_L$. In this case, $\Delta_L$ is just a subspace of
$\mathbb{F}_2^m$ spanned by $\{\varepsilon_i: i\in L\}$, where $\varepsilon_i$ is the $n$-dimensional row vector whose $i$th entry is $1$ and $0$ otherwise.
\begin{example}{\rm
    Let
    $$\Delta_1=\{(0,0,0),(1,0,0),(0,1,0),(0,0,1),(1,1,0),(0,1,1)\}\subseteq \mathbb{F}_2^3.$$ Then $\Delta_1$ is a simplicial complex with maximal elements $(1,1,0)$
    and
    $(0,1,1)$. \\
    Let $$\Delta_2=\{(0,0,0),(1,0,0),(0,0,1),(1,0,1)\}\subseteq \mathbb{F}_2^3.$$ Then $\Delta_2$ is a simplical complex with maximal
    element $(1,0,1)$. Furthermore, $\Delta_2$ is also a linear space over $\mathbb{F}_2$ generated by $\{(1,0,0),(0,0,1)\}$.}
\end{example}
For a set $S$, denote the power set of $S$ by $2^S$.

For $x\in\mathbb{F}_2^m$, define $\chi_x:2^{\mathbb{F}_2^m}\to \mathbb{Z}$ by
\begin{equation*}
    \chi_x(P)=\sum_{y\in P}(-1)^{x\cdot y},
\end{equation*}
If $V$ is a subspace of $\mathbb{F}_2^n$, then we have the following result:
\begin{lemma} \label{lem1}
    Let $V$ be a linear subspace of $\mathbb{F}_2^m$, then we have
    \begin{equation*}
        \chi_x(V)=\left\{\begin{array}{cl}
            |V|, &x\in V^\perp;\\
            0, &x\notin V^\perp.
        \end{array}\right.
    \end{equation*}
\end{lemma}
\begin{proof}
    If $x\in V^\perp$, then $x\cdot y=0$ for all $y\in V$ and $(-1)^{x\cdot y}=1$ for all $y\in V$, thus $\chi_x(V)=|V|$. If $x\notin V^\perp$, then there exsists
    $y_0\in V$ such that $x\cdot y_0=1$. Let
    \begin{equation}
        \begin{aligned}
            &A=\{y\in V:x\cdot y=0\},\\
            &B=\{y\in V:x\cdot y=1\}.
        \end{aligned}
    \end{equation}
    Then both $A$ and $B$ are nonempty. Note that $y_0+B=\{y_0+y:y\in B\}\subseteq A$ and $x_0+A=\{x_0+y:y\in A\}\subseteq B$, thus $|y_0+B|=|B|\leq |A|$ and
    $|x_0+A|=|A|\leq |B|$, hence $|A|=|B|$. Therefore, $\chi_x(V)=\chi_x(A)+\chi_x(B)=|A|-|B|=0$.
\end{proof}

\begin{corollary}\label{chi}
    Let $L\subseteq [m]$, where $m\in \mathbb{N}$. Let $L^c=[m]\backslash L$, then
    \begin{equation*}
            \chi_x(\Delta_{L})=\left\{\begin{array}{cll}
                2^{|L|}&,&x\in\Delta_{L^c};\\
                0&,&x\notin\Delta_{L^c}.
            \end{array}\right.\\
    \end{equation*}
\end{corollary}
\begin{proof}
    Note that $\Delta_L$ is a subspace of $\mathbb{F}_2^m$ spanned by $\{\varepsilon_i: i\in L\}$, then $|\Delta_L|=2^{|L|}$ and $\Delta_L^\perp=\Delta_{L^c}$. Thus
    from Lemma \ref{lem1} we get the conclusion immediately.
\end{proof}
\subsection{The LFSR sequences}
In this subsection, we use the notation $\mathbb{F}=\mathbb{F}_q$. We will introduce some basic concepts about linear feedback shift register(LFSR) sequences or linear
recursive sequences. All the following results are from \cite{golomb2005signal} and
more details about the LFSR sequences can also be found in \cite{golomb2005signal}.

Let $V(\mathbb{F})$ be the set of all infinite sequences with elements in $\mathbb{F}$; that is,
\begin{equation}
    V(\mathbb{F})=\{\mathbf{a}=(a_0,a_1,\cdots):a_i\in \mathbb{F}\}.
\end{equation}\\
Let
\begin{equation*}
    \begin{aligned}
        &\mathbf{a}=(a_0,a_1,a_2,\cdots),\\
        &\mathbf{b}=(b_0,b_1,b_2,\cdots)
    \end{aligned}
\end{equation*}
be two sequences in $V(\mathbb{F})$ and let $c\in V(\mathbb{F})$. We define the addition and the scalar multiplication on $V(\mathbb{F})$ as follows:
\begin{equation*}
    \begin{aligned}
        \mathbf{a}+\mathbf{b}&=(a_0+b_0,a_1+b_1,a_2+b_2,\cdots),\\
        c\mathbf{a}&=(ca_0,ca_1,ca_2,\cdots).
    \end{aligned}
\end{equation*}
It is easy to verify that $V(\mathbb{F})$ is a linear space over $\mathbb{F}$ under this two operations. For $\mathbf{a}=(a_0,a_1,a_2,\cdots)\in V(F)$, we define a
left shift operator $L$ as follows:
$$L(\mathbf{a})=(a_1,a_2,a_3,\cdots).$$
Generally, for any positive integer $i$, we have
\begin{equation*}
    L^i\mathbf{a}=(a_i,a_{i+1},a_{i+2},\cdots).
\end{equation*}
By convention, we write $L^0\mathbf{a}=I\mathbf{a}=\mathbf{a}$, where $I$ is the identity transformation on $V(\mathbb{F})$.
An LFSR sequence is a sequence
$$\mathbf{a}=(a_0,a_1,a_2,\cdots)$$
in $V(\mathbb{F})$ whose elements satisfies the linear recursive relation
\begin{equation} \label{recursive}
    a_{n+k}=\sum_{i=0}^{n-1}c_ia_{i+k}, k=0,1,\cdots,
\end{equation}
where $c_i\in \mathbb{F}$, $i=0,1,\cdots,n-1$. By using the left shift operation $L$, the formula \eqref{recursive} can be written as
$$L^n\mathbf{a}=\sum_{i=0}^{n-1}c_iL^i\mathbf{a},$$
or equivalently,
\begin{equation}\label{poly}
    \left(L^n-\sum_{i=0}^{n-1}c_iL^i\right)\mathbf{a}=\mathbf{0}.
\end{equation}
We wirte
\begin{equation}\label{fx}
        f(x)=x^n-(c_{n-1}x^{n-1}+\cdots+c_1x+c_0),
\end{equation}
then we have $f(L)=L^n-(c_{n-1}L^{n-1}+\cdots+c_1L+c_0I)$ and $f(L)\mathbf{a}=\mathbf{0}$.
From \eqref{poly}, the definition of LFSR sequences (or linear recursive sequences) is equivalent to the following definition.
\begin{definition}{\rm (\cite{golomb2005signal})}
    For any infinite sequence $\mathbf{a}$ in $V(\mathbb{F})$, if there exists a nonzero monic polynomial $f(x)\in \mathbb{F}[x]$ such that
    \begin{equation*}
        f(L)\mathbf{a}=\mathbf{0},
    \end{equation*}
    then $\mathbf{a}$ is called a linear recursive sequence, or equivalently, an LFSR sequence. The polynomial $f(x)$ is called the characteristic polynomial of
    $\mathbf{a}$ over $\mathbb{F}$.
\end{definition}
For any nonzero polynomial $f(x)\in \mathbb{F}[x]$, let $G(f)$ be the set consisting of all sequences in $V(\mathbb{F})$ with
\begin{equation*}
    f(L)\mathbf{a}=\mathbf{0}.
\end{equation*}
Since $f(L)$ is also a linear transformation, $G(f)$ is a subspace of $V(\mathbb{F})$.
For any LFSR sequence $\mathbf{a}=(a_0,a_1,a_2,\cdots)\in G(f)$, where $\deg f=n$, we call the succesive $n$ terms in $\mathbf{a}$ such as
$$(a_k,a_{k+1},\cdots,a_{k+(n-1)}),\quad k\geq0$$ a state, denoted by $s_k$. The state $s_0=(a_0,a_1,\cdots,a_{n-1})$ is called the initial state of $\mathbf{a}$.
\begin{theorem} \label{G(f)}{\rm (\cite{golomb2005signal})}
    Let $f(x)\in \mathbb{F}[x]$ be a monic polynomial of degree $n$. Then $G(f)$ is a linear space of dimension $n$. Hence it contains $q^n$ different sequences. In
    particular, if $q=2$, $G(f)$ contains $2^n$ different binary sequences.
\end{theorem}

For an LFSR sequence $\mathbf{a}$, we define
\begin{equation*}
    A(\mathbf{a})=\{f(x)\in \mathbb{F}[x]:f(L)\mathbf{a}=\mathbf{0}\}.
\end{equation*}

In \cite{golomb2005signal}, the authors proved $A(\mathbf{a})$ is an ideal of $\mathbb{F}[x]$, then from the basic ring theory, there exists an uniquely determined
monic polynomial $m(x)\in A(\mathbf{a})$ whose degree is the lowest in $A(\mathbf{a})$ satisfies the property: for
$f(x)\in \mathbb{F}[x]$, $f(L)\mathbf{a}=\mathbf{0}$ if and only if $m(x)|f(x)$. In other words, $\mathbf{a}\in G(f)$ if and only if $m(x)|f(x)$. This uniquely
determined polynomial $m(x)$ is called the minimal polynomial of $\mathbf{a}$ over $\mathbb{F}$.
\begin{corollary} \label{irr}
    If $f(x)$ is a monic irreducible polynomial in $\mathbb{F}[x]$, then $f(x)$ is the minimal polynomial of any nonzero sequence in $G(f)$.
\end{corollary}
Next we introduce the most important theorem for us about LFSR sequences. This result is a part of Theorem 8.51 in \cite{1996Finite}, here we change the notations in
this theorem.
\begin{theorem}\label{state}{\rm\cite[Theorem 8.51]{1996Finite}}
    Let $\mathbf{a}=(a_0,a_1,a_2,\cdots)$ be an LFSR sequence over $\mathbb{F}$ with minimal polynomial $m(x)$ of degree $n\geq 1$. Then the first $n$ succesive
    states
    \begin{equation*}
        s_0,s_{1},\cdots,s_{n-1}
    \end{equation*}
    are linearly independent over $\mathbb{F}$, where
    \begin{equation*}
        s_i=(a_i,a_{i+1},\cdots,a_{i+n-1}).
    \end{equation*}
\end{theorem}

From Corollary \ref{irr} and Theorem \ref{state}, we have
\begin{corollary} \label{coro1}
    Let $f(x)$ be an irreducible polynomial of degree $n$ over $\mathbb{F}$ and $\mathbf{a}\in G(f)$. Then the first $n$ succesive states
    \begin{equation*}
        s_0,s_{1},\cdots,s_{n-1}
    \end{equation*}
    are linearly independent over $\mathbb{F}$.
\end{corollary}
\section{The weight of codewords in $C_{D^*}$ and $C_{D^c}$}
Suppose that $\mathbb{F}_{2^n}=\mathbb{F}_2(\omega)$, where $\omega$ is a root of $f(x)$ and $f(x)=\sum_{i=0}^na_ix^i$ is an irreducible polynomial over
$\mathbb{F}_2$ of degree $n$. Let $m\in \mathbb{N}$ and let $D_i\subseteq \mathbb{F}_2^m$, $0\leq i \leq n-1$. Assume that
$D=D_0+\omega D_1+\dotsb +\omega^{n-1}D_{n-1} \subseteq\mathbb{F}_{2^n}^m$. Let $D^*=D\backslash\{0\}$. Define a linear code as follows:
$$C_{D^*}=\{(v\cdot d)_{d\in D^*}: v\in \mathbb{F}_{2^n}^m\}.$$
Observe that the map $c_{D^*}:\mathbb{F}_{2^n}^m \to C_{D^*}$ defined by $c_{D^*}(v)=(v\cdot d)_{d\in D^*}$ is a surjective linear
transformation. By definition, $\vert D^* \vert$ is the length of $C_{D^*}$. Assume that $v=\alpha_0+\omega \alpha_1+\cdots +\omega^{n-1}
\alpha_{n-1}\in\mathbb{F}_{2^n}^m$ and $d=d_0+\omega d_1+\cdots+\omega^{n-1}d_{n-1}\in D$, where $\alpha_i\in \mathbb{F}_2^m$ and $d_i
\in D_i$, $1\leq i\leq n-1$. Thus
\begin{equation} \label{vd}
    \begin{aligned}
        v\cdot d &=\alpha_0\cdot d_0\\
                 &+\omega(\alpha_1\cdot d_0+\alpha_0\cdot d_1)\\
                 &+\omega^2(\alpha_2\cdot d_0+\alpha_1\cdot d_1+\alpha_0\cdot d_2)\\
                 &+\cdots\\
                 &+\omega^{n-1}(\alpha_{n-1}\cdot d_0+\alpha_{n-2}\cdot d_1+\cdots+\alpha_0\cdot d_{n-1})\\
                 &+\omega^n(\alpha_{n-1}\cdot d_1+\alpha_{n-2}\cdot d_2+\cdots+\alpha_1\cdot d_{n-1})\\
                 &+\cdots\\
                 &+\omega^{2n-2}\alpha_{n-1}\cdot d_{n-1}\\
                 &=\sum_{k=0}^{2n-2}\omega^k\sum_{\substack{0\leq i,j\leq {n-1},\\i+j=k}}\alpha_i\cdot d_j\\
                 &=\sum_{k=0}^{2n-2}\mu_k\omega^k, \quad\text{where}~\mu_k=\sum_{\substack{0\leq i,j\leq {n-1},\\i+j=k}}\alpha_i\cdot d_j.
    \end{aligned}
\end{equation}
Note that $\omega^k$ has the unique expression
\begin{equation*}
    \omega^k=l_{k,0}+l_{k,1}\omega+l_{k,2}\omega^2+\cdots+l_{k,n-1}\omega^{n-1},\quad k\geq 0.
\end{equation*}
Thus $v\cdot d$ can be simplified into the form as follows:
$$v\cdot d=\eta_0+\omega \eta_1+\cdots +\omega^{n-1}\eta_{n-1},$$
where
$$\eta_i=\beta_{i,0}\cdot d_0+\beta_{i,1}\cdot d_1+\cdots+\beta_{i,n-1}\cdot d_{n-1}$$
and $\beta_{i,j}$ is a linear combination of $\alpha_0,\alpha_1,\cdots,\alpha_{n-1}$ with coefficients in $\mathbb{F}_2$.
Now we caculate $\wt(c_{D^*}(v))$ for $v\in \mathbb{F}_{2^n}^m$:
$$
\begin{aligned}
    \wt(c_{D^*}(v)) &=\wt((\eta_0+\omega \eta_1+\cdots +\omega^{n-1}\eta_{n-1})_{d\in D})\\
    &=\wt\left(\left(\sum_{j=0}^{n-1}\beta_{0,j}\cdot d_0,\omega \sum_{j=0}^{n-1}\beta_{1,j}\cdot d_1,\cdots,
    \omega^{n-1}\sum_{j=0}^{n-1}\beta_{n-1,j}\cdot d_{n-1} \right)_{d_i\in D_i}\right).
\end{aligned}$$
Note that if $u=u_0+\omega u_1+\cdots+\omega^{n-1}u_{n-1}\in \mathbb{F}_{2^n}$ with $u_i\in \mathbb{F}_2,0\leq i\leq n-1$,
then $u=0\Leftrightarrow u_i=0, 0\leq i\leq n-1$.\\
Hence,
\begin{equation} \label{wt}
\begin{aligned}
    \wt(c_{D^*}(v))&=|D|-\frac{1}{2^n}\sum_{d_0\in D_0}\sum_{d_1\in D_1}\cdots\sum_{d_{n-1}\in D_{n-1}}(1+(-1)^{\eta_0})\times
    (1+(-1)^{\eta_1})\times \cdots \\
    &\quad\quad\quad\quad\quad\quad\quad\quad\quad\quad\quad\quad\quad\quad\times(1+(-1)^{\eta_{n-1}})\\
    &=|D|-\frac{1}{2^n}\sum_{\substack{0\leq j_1<j_2<\cdots<j_k\leq n-1,\\0\leq k\leq n}}\sum_{d_0\in D_0}\sum_{d_1\in D_1}\cdots\sum_{d_{n-1}\in D_{n-1}}
    (-1)^{\eta_{j_1}+\eta_{j_2}+\cdots+\eta_{j_k}},
\end{aligned}
\end{equation}
where the sum $\eta_{j_1}+\eta_{j_2}+\cdots+\eta_{j_k}=0$ if $k=0$.

Let $V$ be the set of all elemets of the formal sum $$k_0\alpha_0+k_1\alpha_1+\dotsb +k_{n-1}\alpha_{n-1},~ k_i\in \mathbb{F}_2,~i=0,1\cdots,n-1.$$
Two elements $\sum_{i=0}^{n-1}k_i\alpha_i$ and $\sum_{i=0}^{n-1}k_i'\alpha_i$ in $V$ are equal if $k_i=k_i'$, $\forall~ 0\leq i\leq n-1$.
Define the addition and the scalar multiplication on $V$ as follows:
$$
    \begin{aligned}
        &\sum_{i=0}^{n-1}k_i\alpha_i+\sum_{i=0}^{n-1}k_i'\alpha_i=\sum_{i=0}^{n-1}(k_i+k_i')\alpha_i,\\
        &l\sum_{i=0}^{n-1}k_i\alpha_i=\sum_{i=0}^{n-1}(lk_i)\alpha_i,~\forall~ l\in \mathbb{F}_2.
    \end{aligned}
$$
It is easy to verify that $V$ becomes a linear space over $\mathbb{F}_2$ under the above two operations. We identify the element
$$0\alpha_0+\cdots+0\alpha_{i-1}+1\alpha_i+0\alpha_{i+1}+\cdots+0\alpha_{n-1}$$ with $\alpha_i$, then we can easily verify that $\{\alpha_0,\cdots,\alpha_{n-1}\}$
form a basis of $V$, thus $\dim V=n$.
\begin{remark}{\rm
    The definition of $V$ is different from the subspace of $\mathbb{F}_2^m$ spanned by $\{\alpha_0,\alpha_1,\cdots,\alpha_{n-1}\}$ with $\alpha_i\in\mathbb{F}_2^m$,
    since we just regard the $k_0\alpha_0+k_1\alpha_1+\cdots+k_{n-1}\alpha_{n-1}$ as the formal sum of $\alpha_0,\alpha_1,\dotsb,\alpha_{n-1}$ instead of elements of
    the usual linear combination of $\alpha_0,\alpha_1,\cdots,\alpha_{n-1}$ in $\mathbb{F}_2^m$. However, we can also regard the element in $V$ as an element in
    $\mathbb{F}_2^m$. In order to differentiate it, for $\mathtt{a}\in V$, we let $\tilde{\mathtt{a}}$ denote the element in $\mathbb{F}_2^m$ correspond to
    $\mathtt{a}$.}
\end{remark}

Assume that
\begin{equation*}
    \beta_{i,j}=t_{i,j,0}\alpha_0+t_{i,j,1}\alpha_1+\cdots+t_{i,j,n-1}\alpha_{n-1},
\end{equation*}
we associate each $\eta_i$ with a matrix
$$M_i=\left(\begin{matrix}
    t_{i,0,0} & t_{i,0,1} & \cdots & t_{i,0,n-1}\\
    t_{i,1,0} & t_{i,1,1} & \cdots & t_{i,1,n-1}\\
    \vdots    & \vdots    & \ddots & \vdots     \\
    t_{i,n-1,0} & t_{i,n-1,1} & \cdots & t_{i,n-1,n-1}
\end{matrix}\right)$$
such that the $n$ row vectors of $M_i$ are the coordinates of $\beta_{i,0},\beta_{i,1},\cdots,\beta_{i,{n-1}}$ under the basis
$\{\alpha_0,\cdots,\alpha_{n-1}\}$ respectively if we regard $\beta_{i,j}$ as an element in $V$. In the same manner, let $A_k=\sum_{i+j=k+2}E_{ij}$ be the associated
matrix with $\mu_k$, $0\leq k\leq 2n-2$, where $E_{ij}$ denotes the elementary matrix whose $(i,j)$ entry is $1$ and $0$ otherwise. Then we have the following result
from a trivial observation.
\begin{lemma}\label{M_i}
    Let the notations be the same as above, we have
    \begin{equation*}
        M_i=A_i+\sum_{j=n}^{2n-2}l_{j,i}A_j.
    \end{equation*}
    In particular, the first row vector of $M_i$ is $\varepsilon_{i+1}$, thus $\beta_{i,0}=\alpha_i$, $0\leq i\leq n-1$.
\end{lemma}
\begin{proof}
    For $j\geq n$, $\omega^j=l_{j,0}+l_{j,1}\omega+l_{j,2}\omega^2+\cdots+l_{j,n-1}\omega^{n-1}$. If $l_{ji}=1$ then $\mu_j$ contributes to $\eta_i$. Thus
    $\eta_i=\mu_i+\sum_{j=n}^{2n-2}l_{j,i}\mu_j$, hence $M_i=A_i+\sum_{j=n}^{2n-2}l_{j,i}A_j$.
\end{proof}

Let $V^n$=$\{(\mathtt{a}_0,\cdots,\mathtt{a}_{n-1}):\mathtt{a}_i\in V\}$, define the addition and the scalar multiplication on $V^n$ as follows:
$$\begin{aligned}
    &(\mathtt{a}_0,\cdots,\mathtt{a}_{n-1})+(\mathtt{a}_0',\cdots,\mathtt{a}'_{n-1})=(\mathtt{a}_0+\mathtt{a}'_0,\cdots,\mathtt{a}_{n-1}+\mathtt{a}'_{n-1}),\\
    &k(\mathtt{a}_0,\cdots,\mathtt{a}_{n-1})=(k\mathtt{a}_0,\cdots,k\mathtt{a}_{n-1}),~\forall~k\in \mathbb{F}_2.
\end{aligned}$$
Thus $V^n$ also becomes a linear space over $\mathbb{F}_2$ under the above two operations. Let $$\gamma_i=(\beta_{i,0},\beta_{i,1},\cdots,\beta_{i,n-1}).$$
Then $\gamma_0,\gamma_1,\cdots,\gamma_{n-1}$ are linearly independent in $V^n$ since $\beta_{i,0}=\alpha_i$ by Lemma \ref{M_i}.
Let $W$ be the linear subspace of $V^n$ spanned by $\{\gamma_0,\gamma_1,\cdots,\gamma_{n-1}\}$, note that there exsists a natural bijection from
$\{\eta_{j_1}+\eta_{j_2}+\cdots+\eta_{j_k}:0\leq j_1<j_2<\cdots<j_k\leq n-1,0\leq k\leq n \}$ to $W$ defined by $\eta_{j_1}+\eta_{j_2}+\cdots+\eta_{j_k}\mapsto
\gamma_{j_1}+\gamma_{j_2}+\cdots+\gamma_{j_k}$. Recall that
\begin{equation*}
    \chi_x(P)=\sum_{y\in P}(-1)^{x\cdot y},
\end{equation*}
where $x\in \mathbb{F}_2^m$ and $P\subseteq \mathbb{F}_2^m$. Then from Equation \eqref{wt} we have
\begin{equation} \label{eq1}
    \begin{aligned}
        \wt(c_{D^*}(v))&=|D|-\frac{1}{2^n}\sum_{\substack{0\leq j_1<j_2<\cdots<j_k\leq n-1,\\0\leq k\leq n}}\sum_{d_0\in D_0}\sum_{d_1\in D_1}\cdots\sum_{d_{n-1}
        \in D_{n-1}}(-1)^{\eta_{j_1}+\eta_{j_2}+\cdots+\eta_{j_k}}\\
        &=|D|-\frac{1}{2^n}\sum_{w\in W}\sum_{d_0\in D_0}\sum_{d_1\in D_1}\cdots\sum_{d_{n-1}\in D_{n-1}}(-1)^{\tilde{w_0}\cdot d_0+\tilde{w_1}\cdot d_1+\cdots+
        \tilde{w_{n-1}}\cdot d_{n-1}}\\
        &=|D|-\frac{1}{2^n}\sum_{w\in W}\chi_{\tilde{w_0}}(D_0)\chi_{\tilde{w_1}}(D_1)\cdots \chi_{\tilde{w_{n-1}}}(D_{n-1}),
    \end{aligned}
\end{equation}
where $w=(w_0,w_1,\cdots,w_{n-1})$.

Next we consider the code $C_{D^c}$, where $D^c=\mathbb{F}_{2^n}^m\backslash D$. First we  write $D^c$ as follows:
\begin{equation}\label{eq2}
    \begin{aligned}
        D^{c}=&\left(D_{0}^{c}+\omega \mathbb{F}_{2}^{m}+\cdots+\omega^{n-1} \mathbb{F}_{2}^{m}\right) \bigsqcup \\
        &\left(D_{0}+\omega D_{1}^{c}+\omega^{2} \mathbb{F}_{2}^{m} \cdots+\omega^{n-1} \mathbb{F}_{2}^{m}\right) \bigsqcup \\
        & \vdots \\
        &\left(D_{0}+\omega D_{1}+\cdots+\omega^{n-2} D_{n-2}^{c}+\omega^{n-1} \mathbb{F}_{2}^{m}\right) \bigsqcup \\
        &\left(D_{0}+\omega D_{1}+\cdots+\omega^{n-2} D_{n-2}+\omega^{n-1} D_{n-1}^{c}\right),
    \end{aligned}
\end{equation}
where $\bigsqcup$ indicates disjoint union.

Note that $\chi_{x}(P)+\chi_{x}(P^c)=\chi_x(\mathbb{F}_2^m)=2^m\delta_{0,x}$ for
$x\in\mathbb{F}_2^m,P\subseteq\mathbb{F}_2^m$, where $\delta$ is the Kronecker delta function. Then by Equations \eqref{eq1} and \eqref{eq2}, we have
\begin{equation}\label{Dc}
    \begin{aligned}
        \wt(c_{D^c}(v))&=|D^c|-\frac{1}{2^n}\sum_{w\in W}\left(\chi_{\tilde{w_0}}(D_0^c)\chi_{\tilde{w_1}}(\mathbb{F}_2^m)\cdots \chi_{\tilde{w_{n-1}}}(\mathbb{F}_2^m)
        \right.\\
        &\quad+\chi_{\tilde{w_0}}(D_0)\chi_{\tilde{w_1}}(D_1^c)\chi_{\tilde{w_2}}(\mathbb{F}_2^m)\cdots \chi_{\tilde{w_{n-1}}}(\mathbb{F}_2^m)\\
        &\quad+\cdots\\
        &\quad+\chi_{\tilde{w_0}}(D_0)\chi_{\tilde{w_1}}(D_1)\cdots \chi_{\tilde{w_{n-2}}}(D_{n-2}^c)\chi_{\tilde{w_{n-1}}}(\mathbb{F}_2^m)\\
        &\quad+\left.\chi_{\tilde{w_0}}(D_0)\chi_{\tilde{w_1}}(D_1)\dotsb \chi_{\tilde{w_{n-2}}}(D_{n-2})\chi_{\tilde{w_{n-1}}}(D_{n-1}^{c})\right)\\
        &=|D^c|-\frac{1}{2^n}\sum_{w\in W}\left(2^{nm}\delta_{0,\tilde{w}}-\chi_{\tilde{w_0}}(D_0)\chi_{\tilde{w_1}}(D_1)\cdots \chi_{\tilde{w_{n-1}}}(D_{n-1})
        \right),
    \end{aligned}
\end{equation}
where $w=(w_0,w_1,\cdots,w_{n-1})$ and $\tilde{w}=(\tilde{w_0},\tilde{w_1},\cdots,\tilde{w_{n-1}})$ and $\delta_{0,\tilde{w}}$ is defined by
$\delta_{0,\tilde{w}}=\delta_{0,\tilde{w_0}}\delta_{0,\tilde{w_1}}\cdots\delta_{0,\tilde{w_{n-1}}}$.

In order to illustrate the relation between $\wt(c_{D^*}(v))$ and $\wt(c_{D^c}(v))$, the following result is important.
\begin{proposition}\label{prop1}
    Let $f(x)=x^n+a_{n-1}x^{n-1}+\cdots+a_1x+1$ be an irreducible polynomial over $\mathbb{F}_2$ and $\omega$ a root of $f(x)$. Assume that
    $$\omega^k=l_{k,n-1}\omega^{n-1}+l_{k,n-2}\omega^{n-2}+\cdots+l_{k,1}\omega+l_{k,0},\quad k\geq 0.$$ Let
    $$\mathbf{a}_i=(l_{0,i},l_{1,i},\cdots,l_{k,i},\cdots).$$ Then $\mathbf{a}_i$ is an LFSR sequence and $\mathbf{a}_i\in G(f)$ for $0\leq i\leq n-1$. Furthermore,
    $\{\mathbf{a}_0,\mathbf{a}_1,\cdots,\mathbf{a}_{n-1}\}$ is a basis of $G(f)$.
\end{proposition}
\begin{proof}
    Since $f(\omega)=0$, then $\omega^n=a_{n-1}\omega^{n-1}+a_{n-2}\omega^{n-2}+\cdots+a_1\omega+1$. For all $k\geq 0$, we have
    $$\begin{aligned}
        \omega^{k+1}&=\omega\cdot\omega^k\\
        &=\omega(l_{k,n-1}\omega^{n-1}+l_{k,n-2}\omega^{n-2}+\cdots+l_{k,1}\omega+l_{k,0})\\
        &=l_{k,n-1}(a_{n-1}\omega^{n-1}+a_{n-2}\omega^{n-2}+\cdots+a_1\omega+1)\\
        &\quad +l_{k,n-2}\omega^{n-1}+l_{k,n-3}\omega^{n-2}+\cdots+l_{k,1}\omega^2+l_{k,0}\omega.
    \end{aligned}
    $$
    Thus we get the following relation
    \begin{equation}\label{relation}
        \left\{
            \begin{aligned}
                l_{k+1,n-1}&=a_{n-1}l_{k,n-1}+l_{k,n-2}\\
                l_{k+1,n-2}&=a_{n-2}l_{k,n-1}+l_{k,n-3}\\
                \vdots &\quad \quad \quad \quad \vdots\\
                l_{k+1,1}&=a_1l_{k,n-1}+l_{k,0}\\
                l_{k+1,0}&=l_{k,n-1}.
            \end{aligned}
        \right.
    \end{equation}
    From  Relation \eqref{relation} we have
    \begin{equation}\label{L}
        \left\{
            \begin{aligned}
                L\mathbf{a}_{n-1}&=a_{n-1}\mathbf{a}_{n-1}+\mathbf{a}_{n-2}\\
                L\mathbf{a}_{n-2}&=a_{n-2}\mathbf{a}_{n-1}+\mathbf{a}_{n-3}\\
                \vdots &\quad \quad \quad \quad \vdots\\
                L\mathbf{a}_1&=a_1\mathbf{a}_{n-1}+\mathbf{a}_0\\
                L\mathbf{a}_0&=\mathbf{a}_{n-1}.
            \end{aligned}
        \right.
    \end{equation}
    Then
    \begin{equation*}
        \begin{aligned}
            L^n\mathbf{a}_{n-1}&=L^{n-1}(L\mathbf{a}_{n-1})\\
            &=L^{n-1}(a_{n-1}\mathbf{a}_{n-1}+\mathbf{a}_{n-2})\\
            &=a_{n-1}L^{n-1}\mathbf{a}_{n-1}+L^{n-2}(L\mathbf{a}_{n-2})\\
            &=a_{n-1}L^{n-1}\mathbf{a}_{n-1}+L^{n-2}(a_{n-2}\mathbf{a}_{n-1}+\mathbf{a}_{n-3})\\
            &=a_{n-1}L^{n-1}\mathbf{a}_{n-1}+a_{n-2}L^{n-2}\mathbf{a}_{n-1}+L^{n-3}(L\mathbf{a}_{n-3})\\
            &=\cdots\\
            &=a_{n-1}L^{n-1}\mathbf{a}_{n-1}+a_{n-2}L^{n-2}\mathbf{a}_{n-1}+\cdots+a_1L\mathbf{a}_{n-1}+L\mathbf{a}_0\\
            &=a_{n-1}L^{n-1}\mathbf{a}_{n-1}+a_{n-2}L^{n-2}\mathbf{a}_{n-1}+\cdots+a_1L\mathbf{a}_{n-1}+\mathbf{a}_{n-1}.
        \end{aligned}
    \end{equation*}
    It follows that $f(L)\mathbf{a}_{n-1}=\mathbf{0}$ and then $\mathbf{a}_{n-1}\in G(f)$.
    From Relation \eqref{L} and Theorem \ref{G(f)}, we conclude that $\mathbf{a}_i\in G(f)$ for $0\leq i\leq n-1$. Furthermore, note that the initial state of
    $\mathbf{a}_i$ is $\varepsilon_{i+1}$ for $0\leq i\leq n-1$, then $\mathbf{a}_0,\mathbf{a}_1,\cdots,\mathbf{a}_{n-1}$ are linearly independent over $\mathbb{F}_2$
    and therefore $\{\mathbf{a}_0,\mathbf{a}_1,\cdots,\mathbf{a}_{n-1}\}$ is a basis of $G(f)$.
\end{proof}

\begin{lemma}\label{M}
    Let $\mathcal{M}$ be the set of all $\mathbb{F}_2$-linear combination of $M_0,M_1,\cdots,M_{n-1}$. Then the following hold.\\
    (1). The $n$ row vectors of $M_i$ are the first $n$ succesive states of $\mathbf{a}_i$. \\
    (2). $M$ is invertible for all $M\in \mathcal{M}$.\\
    (3). For all nonzero $w=(w_0,w_1,\cdots,w_{n-1})\in W$, $\{w_0,w_1,\cdots,w_{n-1}\}$ is a basis of $V$.
\end{lemma}
\begin{proof}
    (1). From Lemma \ref{M_i}, we have $M_i=A_i+\sum_{j=n}^{2n-2}l_{j,i}A_j=\sum_{j=0}^{2n-2}l_{j,i}A_j$. Then
    $$M_i=\left(
    \begin{matrix}
        l_{0,i}&l_{1,i}&\cdots&l_{n-1,i}\\
        l_{1,i}&l_{2,i}&\cdots&l_{n,i}\\
        \vdots&\vdots&\ddots&\vdots\\
        l_{n-1,i}&l_{n,i}&\cdots&l_{2n-2,i}
    \end{matrix}\right).
    $$
    Since $\mathbf{a}_i=(l_{0,i},l_{1,i},\cdots,l_{k,i},\cdots)$, we get the conclusion immediately.\\
    (2). For $M\in \mathcal{M}$, let $M=M_{i_i}+M_{i_2}+\cdots+M_{i_r}$ and $\mathbf{a}=\mathbf{a}_{i_1}+\mathbf{a}_{i_2}+\cdots+\mathbf{a}_{i_r}$. From Proposition
    \ref{prop1}, we have $\mathbf{a}\in G(f)$, thus the $n$ row vectors of $M$ are the first $n$ succesive states of $\mathbf{a}$ from (1). Then from Corollary
    \ref{coro1}, the $n$ row vectors of $M$ are linearly independent thus $M$ is invertible for all $M\in \mathcal{M}$.\\
    (3). Let $W^*=W\backslash\{0\}$. For all $w=(w_0,w_1,\cdots,w_{n-1})\in W^*$, assume $w=\gamma_{i_1}+\gamma_{i_2}+\cdots+\gamma_{i_r}$. Note that
    $M=M_{i_i}+M_{i_2}+\cdots+M_{i_r}$ is invertible, then $w_0,w_1,\cdots,w_{n-1}$ are linearly independent in $V$ over $\mathbb{F}_2$ since the $n$ row vectors of
    $M$ are the coordinates of $w_0,w_1,\cdots,w_{n-1}$ under the basis $\{\alpha_0,\alpha_1,\cdots,\alpha_{n-1}\}$ respectively, which implies $\{w_0,w_1,\cdots,w_{n-1}\}$
    is a basis of $V$ for all $w\in W^*$.\\
\end{proof}
\begin{lemma}\label{im}
    Let the notations be the same as above. Then
    \begin{equation*}
        \wt(c_{D^c}(v))+\wt(c_{D^*}(v))=(2^n-1)\times 2^{n(m-1)}(1-\delta_{0,v}).
    \end{equation*}
\end{lemma}
\begin{proof}
    From Lemma \ref{M}, $\{w_0,w_1,\cdots,w_{n-1}\}$ is a basis of $V$ for all nonzero $w=(w_0,w_1,\cdots,w_{n-1})\in W$. Thus for
    $v=\alpha_0+\omega\alpha_1+\cdots+\omega^{n-1}\alpha_{n-1}\in\mathbb{F}_{2^n}^m$, where $\alpha_i\in \mathbb{F}_2^m$, we have
    \begin{equation*}
        \begin{aligned}
            v=0&\Longleftrightarrow\alpha_0=\alpha_1=\cdots=\alpha_{n-1}=0~\text{in}~ \mathbb{F}_2^m\\
            &\Longleftrightarrow \tilde{w_0}=\tilde{w_1}=\cdots=\tilde{w_{n-1}}=0~\text{for}~w=(w_0,w_1,\cdots,w_{n-1})\in W^*,
        \end{aligned}
    \end{equation*}
    thus $\delta_{0,\tilde{w}}=\delta_{0,v}$ for all $w\in W^*$. Now from \eqref{Dc} we have
    \begin{equation*}
        \begin{aligned}
            \wt(c_{D^c}(v))&=|D^c|-\frac{1}{2^n}\sum_{w\in W}\left(2^{nm}\delta_{0,\tilde{w}}-\chi_{\tilde{w_0}}(D_0)\chi_{\tilde{w_1}}(D_1)\cdots
            \chi_{\tilde{w_{n-1}}}(D_{n-1})\right)\\
            &=\frac{2^n-1}{2^n}|D^c|-\frac{1}{2^n}\sum_{w\in W^*}\left(2^{nm}\delta_{0,\tilde{w}}-\chi_{\tilde{w_0}}(D_0)\chi_{\tilde{w_1}}(D_1)\cdots
            \chi_{{\tilde{w_{n-1}}}}(D_{n-1})\right)\\
            &=\frac{2^n-1}{2^n}|D^c|-2^{n(m-1)}\delta_{0,v}+\frac{1}{2^n}\sum_{w\in W^*}\chi_{\tilde{w_0}}(D_0)\chi_{\tilde{w_1}}(D_1)\cdots
            \chi_{\tilde{w_{n-1}}}(D_{n-1}).
        \end{aligned}
    \end{equation*}
    On the other hand, from \eqref{eq1}, we have
    \begin{equation*}
        \begin{aligned}
            \wt(c_{D^*}(v))&=|D|-\frac{1}{2^n}\sum_{w\in W}\chi_{\tilde{w_0}}(D_0)\chi_{\tilde{w_1}}(D_1)\cdots \chi_{\tilde{w_{n-1}}}(D_{n-1})\\
            &=\frac{2^n-1}{2^n}|D|-\frac{1}{2^n}\sum_{w\in W^*}\chi_{\tilde{w_0}}(D_0)\chi_{\tilde{w_1}}(D_1)\cdots \chi_{\tilde{w_{n-1}}}(D_{n-1}).
        \end{aligned}
    \end{equation*}
    Thus
    \begin{equation*}
            \wt(c_{D^c}(v))+\wt(c_{D^*}(v))=(2^n-1)\times 2^{n(m-1)}(1-\delta_{0,v}).
    \end{equation*}
\end{proof}
We give an example to illustrate the above results in this section.
\begin{example} \label{ex1}
    {\rm Let $\mathbb{F}_8=\mathbb{F}_2(\omega)$, where $\omega$ is a root of $f(x)=x^3+x+1$. Then
    \begin{equation*}
        \begin{aligned}
            \omega^3&=\omega+1,\\
            \omega^4&=\omega^2+\omega,\\
            \omega^5&=\omega^2+\omega+1,\\
            \omega^6&=\omega^2+1,\\
            \omega^7&=1.
        \end{aligned}
    \end{equation*}
    Let $D=D_0+\omega D_1+\omega^2 D_2\subseteq \mathbb{F}_{2^3}^m$ and $v=\alpha_0+\omega\alpha_1+\omega^2\alpha_2$, $d=d_0+\omega d_1+\omega^2 d_2$ where
    $d_i\in D_i$. Then
    \begin{equation*}
        \begin{aligned}
            v\cdot d &=\alpha_0\cdot d_0\\
                     &\quad+\omega(\alpha_1\cdot d_0+\alpha_0\cdot d_1)\\
                     &\quad+\omega^2(\alpha_2\cdot d_0+\alpha_1\cdot d_1+\alpha_0\cdot d_2)\\
                     &\quad+(\omega+1)(\alpha_2\cdot d_1+\alpha_1\cdot d_2)\\
                     &\quad+(\omega^2+\omega)\alpha_2\cdot d_2\\
                     &=\alpha_0\cdot d_0+\alpha_2\cdot d_1+\alpha_1\cdot d_2\\
                     &\quad+\omega(\alpha_1\cdot d_0+(\alpha_0+\alpha_2)\cdot d_1+(\alpha_1+\alpha_2)\cdot d_2)\\
                     &\quad+\omega^2(\alpha_2\cdot d_0+\alpha_1\cdot d_1+(\alpha_0+\alpha_2)\cdot d_2).
        \end{aligned}
    \end{equation*}
    Thus
    \begin{equation*}
        \begin{aligned}
            \eta_0&=\alpha_0\cdot d_0+\alpha_2\cdot d_1+\alpha_1\cdot d_2,\\
            \eta_1&=\alpha_1\cdot d_0+(\alpha_0+\alpha_2)\cdot d_1+(\alpha_1+\alpha_2)\cdot d_2,\\
            \eta_2&=\alpha_2\cdot d_0+\alpha_1\cdot d_1+(\alpha_0+\alpha_2)\cdot d_2.
        \end{aligned}
    \end{equation*}
    Then we have
    $$
    M_0=\left(\begin{matrix}
        1&0&0\\
        0&0&1\\
        0&1&0
    \end{matrix}\right),\quad
    M_1=\left(\begin{matrix}
        0&1&0\\
        1&0&1\\
        0&1&1
    \end{matrix}\right),\quad
    M_2=\left(\begin{matrix}
        0&0&1\\
        0&1&0\\
        1&0&1
    \end{matrix}\right).
    $$
    Note that
    $$
    A_0=\left(\begin{matrix}
        1&0&0\\
        0&0&0\\
        0&0&0
    \end{matrix}\right),\quad
    A_1=\left(\begin{matrix}
        0&1&0\\
        1&0&0\\
        0&0&0
    \end{matrix}\right),\quad
    A_2=\left(\begin{matrix}
        0&0&1\\
        0&1&0\\
        1&0&0
    \end{matrix}\right),
    $$
    $$
    A_3=\left(\begin{matrix}
        0&0&0\\
        0&0&1\\
        0&1&0
    \end{matrix}\right),\quad
    A_4=\left(\begin{matrix}
        0&0&0\\
        0&0&0\\
        0&0&1
    \end{matrix}\right).
    $$
    Then we get $M_0=A_0+A_3$, $M_1=A_1+A_3+A_4$, $M_2=A_2+A_4$. In particular, let
    \begin{equation*}
        \begin{aligned}
            \mathbf{a}_0&=(1001011\cdots),\\
            \mathbf{a}_1&=(0101110\cdots),\\
            \mathbf{a}_2&=(0010111\cdots),
        \end{aligned}
    \end{equation*}
    then $\mathbf{a}_0,\mathbf{a}_1,\mathbf{a}_2\in G(f)$ and the three row vectors of $M_i$ is the first three states of $\mathbf{a}_i$ and invertible in
    $M_2(\mathbb{F}_2)$.

    Furthermore, let $$C_{D^*}=\{(v\cdot d)_{d\in D^*}: v\in \mathbb{F}_{2^3}^m\}.$$ Then we have the following from \cite{sagar2022linear}
    \begin{equation*}
        \begin{aligned}
           \wt(c_{D^*}(v))&=|D|-\frac{1}{8}\sum_{d_0\in D_0}\sum_{d_1\in D_1}\sum_{d_2\in D_2}\left(1+(-1)^{\alpha_0\cdot d_0+\alpha_2\cdot d_1+\alpha_1\cdot d_2}
            \right)\times\\
            &\quad\left(1+(-1)^{\alpha_1\cdot d_0+(\alpha_0+\alpha_2)\cdot d_1+(\alpha_1+\alpha_2)\cdot d_2}\right)\times
            \left(1+(-1)^{\alpha_2\cdot d_0+\alpha_1\cdot d_1+(\alpha_0+\alpha_2)\cdot d_2}\right)\\
            &=\frac{7}{8}|D|-\frac{1}{8}\left(\chi_{\alpha_0}(D_0)\chi_{\alpha_2}(D_1)\chi_{\alpha_1}(D_3)+\right.\\
            &\quad\chi_{\alpha_1}(D_0)\chi_{\alpha_0+\alpha_2}(D_1)\chi_{\alpha_1+\alpha_2}(D_2)+\\
            &\quad\chi_{\alpha_2}(D_0)\chi_{\alpha_1}(D_1)\chi_{\alpha_1+\alpha_2}(D_2)+\chi_{\alpha_0+\alpha_1}(D_0)\chi_{\alpha_0}(D_1)\chi_{\alpha_2}(D_2)+\\
            &\quad\chi_{\alpha_0+\alpha_2}(D_0)\chi_{\alpha_1+\alpha_2}(D_1)\chi_{\alpha_0+\alpha_1+\alpha_2}(D_2)+\\
            &\quad\chi_{\alpha_1+\alpha_2}(D_0)\chi_{\alpha_0+\alpha_1+\alpha_2}(D_1)\chi_{\alpha_0+\alpha_1}(D_2)+\\
            &\quad\chi_{\alpha_0+\alpha_1+\alpha_2}(D_0)\chi_{\alpha_0+\alpha_1}(D_1)\chi_{\alpha_0}(D_2)\left.\right).
        \end{aligned}
    \end{equation*}
    \begin{equation*}
        \begin{aligned}
            \wt(c_{D^c}(v))&=\frac{7}{8}(|D^c|-2^{3m}\delta_{0,v})+\frac{1}{8}\left(\chi_{\alpha_0}(D_0)\chi_{\alpha_2}(D_1)\chi_{\alpha_1}(D_3)+\right.\\
            &\quad\chi_{\alpha_1}(D_0)\chi_{\alpha_0+\alpha_2}(D_1)\chi_{\alpha_1+\alpha_2}(D_2)+\\
            &\quad\chi_{\alpha_2}(D_0)\chi_{\alpha_1}(D_1)\chi_{\alpha_1+\alpha_2}(D_2)+\chi_{\alpha_0+\alpha_1}(D_0)\chi_{\alpha_0}(D_1)\chi_{\alpha_2}(D_2)+\\
            &\quad\chi_{\alpha_0+\alpha_2}(D_0)\chi_{\alpha_1+\alpha_2}(D_1)\chi_{\alpha_0+\alpha_1+\alpha_2}(D_2)+\\
            &\quad\chi_{\alpha_1+\alpha_2}(D_0)\chi_{\alpha_0+\alpha_1+\alpha_2}(D_1)\chi_{\alpha_0+\alpha_1}(D_2)+\\
            &\quad\chi_{\alpha_0+\alpha_1+\alpha_2}(D_0)\chi_{\alpha_0+\alpha_1}(D_1)\chi_{\alpha_0}(D_2)\left.\right).
        \end{aligned}
    \end{equation*}
    Thus we have
    \begin{equation*}
        \wt(c_{D^c}(v))+\wt(c_{D^*}(v))=7\times 2^{3(m-1)}(1-\delta_{0,v}).
    \end{equation*}}
\end{example}
\section{The parameters of $C_{D^*}$ and $C_{D^c}$}
When we choose $D_i=\Delta_{L_i}$ be a simplicial complex generated by $L_i$ where $L_i\subseteq [m]$, the weight of codewords in $C_{D^*}$ and $C_{D^c}$ can be easily
calculated, in order to see it, for $Y\subseteq[m]$, we define the Boolean function $\varphi(\cdot|Y)$: $\mathbb{F}_2^m\to \mathbb{F}_2$ as
\begin{equation*}
    \varphi(x|Y)=\prod_{i \in Y}\left(1-x_{i}\right)=\left\{\begin{array}{ll}
        1, & \text { if } \operatorname{Supp}(x) \cap Y=\emptyset ;\\
        0, & \text { if } \operatorname{Supp}(x) \cap Y \neq \emptyset.
        \end{array}\right.
\end{equation*}
Note that $\Delta_{L_i}$ is also a linear subspace of $\mathbb{F}_2^m$ and $\Delta_{L_i}^\perp =\Delta_{L_i^c}$. Thus from Corollary \ref{chi}, we have
\begin{equation*}
    \begin{aligned}
        \chi_x(\Delta_{L_i})&=\left\{\begin{array}{ccl}
            2^{|L_i|}&,&x\in\Delta_{L_i^c};\\
            0&,&x\notin\Delta_{L_i^c}.
        \end{array}\right.\\
        &=\left\{\begin{array}{ccl}
            2^{|L_i|}&,&\operatorname{Supp}(x) \cap L_i=\emptyset;\\
            0&,&\operatorname{Supp}(x) \cap L_i\neq\emptyset.
        \end{array}\right.\\
        &=2^{|L_i|}\varphi(x|L_i).
    \end{aligned}
\end{equation*}
Then from Equation \eqref{eq1} we have
\begin{equation}\label{eq3}
    \begin{aligned}
        \wt(c_{D^*}(v))&=|D|-\frac{1}{2^n}\sum_{w\in W}\chi_{\tilde{w_0}}(D_0)\chi_{\tilde{w_1}}(D_1)\cdots \chi_{\tilde{w_{n-1}}}(D_{n-1})\\
        &=2^{\sum_{i=0}^{n-1}|L_i|}-2^{\sum_{i=0}^{n-1}|L_i|-n}\sum_{w\in W}\varphi(\tilde{w_0}|L_0)\varphi(\tilde{w_1}|L_1)\cdots\varphi(\tilde{w_{n-1}}|L_{n-1})\\
        &=2^{\sum_{i=0}^{n-1}|L_i|-n}\left(2^n-\sum_{w\in W}\varphi(\tilde{w_0}|L_0)\varphi(\tilde{w_1}|L_1)\cdots\varphi(\tilde{w_{n-1}}|L_{n-1})\right).
    \end{aligned}
\end{equation}

Let $$g(w)=\varphi(\tilde{w_0}|L_0)\varphi(\tilde{w_1}|L_1)\cdots\varphi(\tilde{w_{n-1}}|L_{n-1})$$
and let $$\theta=\sum_{w\in W}g(w).$$ Then we only need to determin the value of $\theta$. Since $g(w)$ can only be $0$ or $1$,
$\theta=|\{w\in W:g(w)=1\}|$.
\begin{lemma} \label{value}
    Let $\hat{W}=\{w\in W:g(w)=1\}$, then $\hat{W}$ is a subspace of $W$. Thus all the possible value of $\theta$ can only be $2^i,~0\leq i\leq n$.
\end{lemma}
\begin{proof}
    Note that $g(0)=1$ and for $w=(w_0,w_1,\cdots,c_{n-1}),w'=(w_0',w_1',\cdots,w_{n-1}')\in W$ such that $g(w)=g(w')=1$, then 
    $\varphi(\tilde{w_i}|L_i)=\varphi(\tilde{w_i'}|L_i)=1$ for all $0\leq i\leq n-1$, thus
    $\operatorname{Supp}(\tilde{w_i})\cap L_i=\operatorname{Supp}(\tilde{w_i'})\cap L_i=\emptyset$
    for all $0\leq i\leq n-1$. Observe that $\operatorname{Supp}(\tilde{w_i}+\tilde{w_i'})\subseteq \operatorname{Supp}(\tilde{w_i})\cup
    \operatorname{Supp}(\tilde{w_i'})$, it follows that $\operatorname{Supp}(\tilde{w_i}+\tilde{w_i'})\cap L_i=\emptyset$ for all $0\leq i\leq n-1$. Then we have 
    $g(w+w')=1$ and $w+w'\in \hat{W}$, hence $\hat{W}$ is a subspace of $W$, which implies that all the possible value of $\theta=|\hat{W}|$ can only be 
    $2^i,0\leq i\leq n$.
\end{proof}

Let $V_0=\{\mathtt{a}=k_0\alpha_0+k_1\alpha_1+\cdots+k_{n-1}\alpha_{n-1}\in V:k_0=0\}$. For $1\leq l\leq n-1$, let $W_l=\{w=(w_0,w_1,\cdots,w_{n-1})\in W:
w_i\in V_0,\forall~i\leq l-1\}$. Then $W_l$ is a subspace of $W$. Note that the first column of $M_l$ is $\varepsilon_{l+1}$, which implies that
$\gamma_l\in W_l$ for $1\leq i\leq n-1$ and $\gamma_l\notin W_{l+1}$ for $0\leq i\leq n-2$. Therefore $0\subsetneqq W_{n-1}\subsetneqq
W_{n-2}\subsetneqq \cdots\subsetneqq W_1\subsetneqq W$ and then $0<\dim W_{n-1}<\dim W_{n-2}<\cdots<\dim W_1<n$, thus $\dim W_l=n-l$. Since
$\gamma_l,\gamma_{l+1},\cdots,\gamma_{n-1}\in W_l$, we have $W_l$ is a subspace of $W$ generated by $\{\gamma_l,\gamma_{l+1},\cdots,\gamma_{n-1}\}$.

Now we can introduce our main results about $C_{D^*}$ and $C_{D^c}$.
\begin{proposition}\label{th1}
    Consider the defining set $D=\Delta_{L_0}+\omega\Delta_{L_1}+\cdots+\omega^{n-1}\Delta_{L_{n-1}}\subseteq \mathbb{F}_{2^n}^m$, where
    $m\in \mathbb{N}$ and $L_i$ are nonempty subsets of $[m]$. Let $R_i=L_i\backslash\cup_{j\neq i}L_j$ and suppose that $R_0,R_1,\cdots,R_{n-2}$
    are nonempty. Then the code $C_{D^*}$ is an $n$-weight linear code over $\mathbb{F}_{2^n}$ of length $2^{\sum_{i=0}^{n-1}|L_i|}-1$,
    dimension $\left|\cup_{i=0}^{n-1}L_i\right|$ and distance $2^{\sum_{i=0}^{n-1}|L_i|-1}$. In particular, $C_{D^*}$ have nonzero codewords of
    weights $2^{\sum_{i=0}^{n-1}|L_i|-n}(2^n-2^i)$, $0\leq i \leq n-1$.
\end{proposition}
\begin{proof}
    Observe the length of the code $C_{D^*}$ is $2^{\sum_{i=0}^{n-1}|L_i|}-1$. Consider the map $c_{D^*}:\mathbb{F}_{2^n}^m \to C_{D^*}$
    defined by $c_{D^*}(v)=(v\cdot d)_{d\in D^*}$ which is a surjective linear transformation and$$
        \begin{aligned}
            \ker(c_{D^*})&=\{v=(v_i)\in \mathbb{F}_{2^n}^m : v\cdot d=0,\forall~d\in D^*\}\\
                         &=\{v=(v_i)\in \mathbb{F}_{2^n}^m : v_i=0~\text{for}~i\in \cup_{i=0}^{n-1}L_i\}.\\
        \end{aligned}$$
    Therefore, $\left|\ker(c_{D^*})\right|=2^{n\left(m-\left|\cup_{i=0}^{n-1}L_i\right|\right)}$. By the first isomorphism theorem of groups, we have
    $|C_{D^*}|=|\mathbb{F}_{2^n}^m|/\left|\ker(c_{D^*})\right|=2^{n\left|\cup_{i=0}^{n-1}L_i\right|}$.
    Hence $\dim(C_{D^*})=\left|\cup_{i=0}^{n-1}L_i\right|$.

    Lemma \ref{value} has showed that all the possible value of $\theta$ can only be $2^i,~0\leq i\leq n$, now we show that $\theta$ can equal to each
    $2^i,~0\leq i\leq n$ when we set different $v$ under the assumption $R_0,R_1,\cdots,R_{n-2}$ are nonempty.

    Now we set $\alpha_1=\alpha_2=\cdots=\alpha_{n-1}=(0,0,\cdots,0)$, then $v=\alpha_0$.

    If $\alpha_0=(0,0,\cdots,0)$, we have $\operatorname{Supp}(\alpha_0)\cap L_j=\emptyset$ for all $0\leq j\leq n-1$ and $\varphi(\alpha_0|L_j)=1$ for all
    $0\leq j\leq n-1$, then $g(w)=1$ for all $w\in W$, thus $\theta=2^n$.

    If $\operatorname{Supp}(\alpha_0)=\{i_0,i_1,\cdots,i_{l-1}\}$ where $i_k\in R_k$, $0\leq k\leq l-1\leq n-2$. Then $\operatorname{Supp}(\alpha_0)
    \cap L_j\neq\emptyset$ for $0\leq j\leq l-1$ and $\operatorname{Supp}(\alpha_0)\cap L_j=\emptyset$ for $l\leq j\leq n-1$. Thus $\varphi(\alpha_0|L_j)=0$ for
    $0\leq j\leq l-1$ and $\varphi(\alpha_0|L_j)=1$ for $l\leq j\leq n-1$. Then $g(w)$ is equal to 1 for all $w\in W_l$ and 0 for all $w\notin W_l$, consequently
    $\theta=|W_l|=2^{n-l}$, $1\leq l\leq n-1$.

    If $\alpha_0=(1,1,\cdots,1)$, we have $\operatorname{Supp}(\alpha_0)\cap L_j\neq\emptyset$ for all $0\leq j\leq n-1$ and $\varphi(\alpha_0|L_j)=0$ for all
    $0\leq j\leq n-1$, then $g(w)=1$ for all nonzero $w\in W$, thus $\theta=1$.

    For $\theta=2^i$, $0\leq i\leq n$, we have
    \begin{equation*}
        \begin{aligned}
            \wt(c_{D^*}(v))&=2^{\sum_{i=0}^{n-1}|L_i|-n}\left(2^n-\sum_{w\in W}\varphi(\tilde{w_0}|L_0)\varphi(\tilde{w_1}|L_1)\cdots\varphi(\tilde{w_{n-1}}|L_{n-1})
            \right)\\
            &=2^{\sum_{i=0}^{n-1}|L_i|-n}(2^n-\theta)\\
            &=2^{\sum_{i=0}^{n-1}|L_i|-n}(2^n-2^i).
        \end{aligned}
    \end{equation*}
    Thus the distance of $C_{D^*}$ is $2^{\sum_{i=0}^{n-1}|L_i|-n}(2^n-2^{n-1})=2^{\sum_{i=0}^{n-1}|L_i|-1}$.
\end{proof}
\begin{remark}
    {\rm The cases $n=1,2,3$ in Proposition $\ref{th1}$ are Example $3.4$ in \cite{hyun2020infinite}, Proposition $4.2$ in \cite{wu2022quaternary} and Proposition $3.3$
    in \cite{sagar2022linear} respectively.}
\end{remark}

\begin{example}
    {\rm Set $n=3,m=4$ and $L_0=\{1,2\},L_1=\{2,3\},L_2=\{3,4\}\subseteq [4]=\{1,2,3,4\}$. Consider the definition set $D=\Delta_{L_0}+\omega\Delta_{L_1}+\omega^2\Delta_{L_2}
    \subset \mathbb{F}_8^4$, where
    \begin{equation*}
        \begin{aligned}
            \Delta_{L_0}&=\{(0,0,0,0),(1,0,0,0),(0,1,0,0),(1,1,0,0)\},\\
            \Delta_{L_1}&=\{(0,0,0,0),(0,1,0,0),(0,0,1,0),(0,1,1,0)\},\\
            \Delta_{L_2}&=\{(0,0,0,0),(0,0,1,0),(0,0,0,1),(0,0,1,1)\}.
        \end{aligned}
    \end{equation*}
    Then $C_{D^*}$ is a $[63,4,32]$ linear code over $\mathbb{F}_8$ which is not distance optimal since the optimal distance of a $[63,4]$ linear code is $54$ by
    \cite{codetables}.}
\end{example}
If all the $n$ subsets in Proposition \ref{th1} are nonempty and equal then we have the following result.
\begin{theorem}\label{th3.2}
    Let $L$ be an nonempty subset of $[m]$. Consider the defining set $D=\Delta_L+\omega\Delta_L+\cdots+\omega^{n-1}\Delta_L\subseteq \mathbb{F}_{2^n}^m$.
    Then the code $C_{D^*}$ is a 1-weight linear code over $\mathbb{F}_{2^n}^m$ of length $2^{n|L|}-1$, dimension $|L|$ and distance $(2^n-1)\times 2^{n(|L|-1)}$.
    In particular, it is a minimal code. Moreover, $C_{D^*}$ is a Griesmer code.
\end{theorem}
\begin{proof}
    From \eqref{eq3}, we have
    \begin{equation*}
        \wt(c_{D^*}(v))=2^{n(|L|-1)}\left(2^n-\sum_{w\in W}\varphi(\tilde{w_0}|L)\varphi(\tilde{w_1}|L)\dotsb \varphi(\tilde{w_{n-1}}|L)\right).
    \end{equation*}
    If there exists a $w=(w_0,\cdots,w_{n-1})\in W^*$
    such that
    \begin{equation*}
        g(w)=\varphi(\tilde{w_0}|L)\varphi(\tilde{w_1}|L)\cdots \varphi(\tilde{w_{n-1}}|L)=1,
    \end{equation*}
    then we have $\varphi(\tilde{w_i}|L)=1$ for all $0\leq i\leq n-1$ and thus $\varphi(\tilde{\mathtt{a}}|L)=1$ for all $\mathtt{a}\in V$ since
    $\{w_0,w_1,\cdots,w_{n-1}\}$ is a basis of $V$ by Lemma \ref{M}.
    Then we have $g(w)=1$ for all $w\in W$. Thus the value of $\theta=|\hat{W}|$ can only be $1$ or $2^n$, then the minimal distance of $C_{D^*}$ is
    $(2^n-1)\times 2^{n(|L|-1)}$.

    From Proposition \ref{th1} we have $\dim C_{D^*}=|L|$, then
    \begin{equation*}
        \begin{aligned}
            \sum_{i=0}^{|L|-1}\left\lceil \frac{d}{q^i}\right\rceil&=\sum_{i=0}^{|L|-1}\frac{(2^n-1)\cdot 2^{n(|L|-1)}}{2^{ni}}\\
            &=\sum_{i=0}^{|L|-1}(2^n-1)\cdot 2^{n(|L|-1-i)}\\
            &=\sum_{i=0}^{|L|-1}(2^n-1)\cdot 2^{ni}\\
            &=2^{n|L|}-1.
        \end{aligned}
    \end{equation*}
    Thus $C_{D^*}$ is a Griesmer code.
\end{proof}
\begin{remark}
    {\rm The case $n=3$ in Theorem $\ref{th3.2}$ is Theorem $3.4$ in \cite{sagar2022linear}. }
\end{remark}
\begin{example}
    {\rm Set $n=2,m=4$ and $L=\{1,2,3\}\subseteq [4]=\{1,2,3,4\}$. Consider the defining set $D=\Delta_L+\omega\Delta_L\subseteq \mathbb{F}_4^4$, where
    \begin{equation*}
        \begin{aligned}
            \Delta_L=&\{(0,0,0,0),(1,0,0,0),(0,1,0,0),(0,0,1,0),\\
            &(1,1,0,0),(1,0,1,0),(0,1,1,0),(1,1,1,0)\}.
        \end{aligned}
    \end{equation*}
    Then the code $C_{D^*}$ is a $[63,3,48]$ 1-weight linear code over $\mathbb{F}_4$ which is distance optimal by \cite{codetables}.}
\end{example}
\begin{theorem}\label{th3.3}
    Suppose $L_i$ be nonempty subsets of $[m]$ such that at least one subset is proper where $m\in\mathbb{N}$. Let $R_i=L_i\backslash\cup_{j\neq i}L_j$ and suppose
    that $R_0,R_1,\cdots,R_{n-2}$ are nonempty. Let $D=\Delta_{L_0}+\omega\Delta_{L_1}+\cdots+\omega^{n-1}\Delta_{L_{n-1}}\subseteq \mathbb{F}_{2^n}^m$. Then the
    code $C_{D^c}$ is a $(n+1)$-weight linear code over $\mathbb{F}_{2^n}$ of length $2^{nm}-2^{\sum_{i=0}^{n-1}|L_i|}$, dimension $m$ and distance
    $(2^n-1)\times(2^{n(m-1)}-2^{\sum_{i=0}^{n-1}|L_i|-n})$. In particular, $C_{D^c}$ have nonzero codewords of weights
    $(2^n-1)\times2^{n(m-1)}-2^{\sum_{i=0}^{n-1}|L_i|-n}(2^n-2^i)$, $0\leq i \leq n$.
    Furthermore, $C_{D^c}$ is a Griesmer code hence it is distance optimal. In fact, it is a minimal code if
    $\sum_{i=0}^{n-1}|L_i|\leq nm-(n+1)$.
\end{theorem}
\begin{proof}
    For nonzero $v\in \mathbb{F}_{2^n}^m$, from Lemma \ref{im} we have
    \begin{equation*}
        \wt(c_{D^c}(v))+\wt(c_{D^*}(v))=(2^n-1)\times 2^{n(m-1)}.
    \end{equation*}
    From Theorem \ref{th1}, $\wt(c_{D^*}(v))\leq 2^{\sum_{i=0}^{n-1}| L_i|-n}(2^n-1)$, thus for nonzero $v\in \mathbb{F}_{2^n}^m$, we have
    \begin{equation*}
        \begin{aligned}
            \wt(c_{D^c}(v))&\geq (2^n-1)\times 2^{n(m-1)}-2^{\sum_{i=0}^{n-1}|L_i|-n}(2^n-1)\\
            &=\frac{2^n-1}{2^n}\left(2^{nm}-2^{\sum_{i=0}^{n-1}|L_i|}\right)\\
            &>0
        \end{aligned}
    \end{equation*}
    since $L_i$ is a proper subset of $[m]$ for some $i$. Thus $\ker(c_{D^c})=0$ and then $\dim(C_{D^c})=m$.
    Now we have
    \begin{equation*}
        \begin{aligned}
            \sum_{i=0}^{m-1}\left\lceil \frac{d}{q^i}\right\rceil&=\sum_{i=0}^{m-1}\left\lceil \frac{(2^n-1)\times(2^{n(m-1)}-2^{\sum_{i=0}^{n-1}|L_i|-n})}{2^{ni}}
            \right\rceil\\
            &=\sum_{i=0}^{m-1}(2^n-1)\cdot 2^{n(m-1)-ni}-\sum_{i=0}^{m-1}\left\lfloor\frac{(2^n-1)\cdot 2^{\sum_{i=0}^{n-1}|L_i|-n}}{2^{ni}}\right\rfloor\\
            &=\left(2^{nm}-1\right)-\sum_{i=0}^{m-1}\left\lfloor(2^n-1)\cdot 2^{\sum_{i=0}^{n-1}|L_i|-n-ni}\right\rfloor.
        \end{aligned}
    \end{equation*}
    If $\sum_{i=0}^{n-1}|L_i|\equiv l \mod{n}$, $0\leq l\leq n-1$, let $k=\frac{\sum_{i=0}^{n-1}|L_i|-l}{n}$, then
    \begin{equation*}
        \begin{aligned}
            &\quad\sum_{i=0}^{m-1}\left\lfloor(2^n-1)\cdot 2^{\sum_{i=0}^{n-1}|L_i|-n-ni}\right\rfloor\\
            &=\sum_{i=0}^{k-1}\left\lfloor(2^n-1)\cdot 2^{\sum_{i=0}^{n-1}|L_i|-n-ni}\right\rfloor+(2^l-1)\\
            &=\sum_{i=0}^{k-1}(2^n-1)\cdot 2^{l+ni}+(2^l-1)\\
            &=2^l\cdot \left(2^{\sum_{i=0}^{n-1}|L_i|-l}-1\right)+(2^l-1)\\
            &=2^{\sum_{i=0}^{n-1}|L_i|}-1.
        \end{aligned}
    \end{equation*}
    Thus
    \begin{equation*}
        \begin{aligned}
            \sum_{i=0}^{m-1}\left\lceil \frac{d}{q^i}\right\rceil&=\left(2^{nm}-1\right)-\sum_{i=0}^{m-1}\left\lfloor(2^n-1)\cdot 2^{\sum_{i=0}^{n-1}-n-ni}
            \right\rfloor\\
            &=\left(2^{nm}-1\right)-\left(2^{\sum_{i=0}^{n-1}|L_i|}-1\right)\\
            &=2^{nm}-2^{\sum_{i=0}^{n-1}|L_i|}.
        \end{aligned}
    \end{equation*}
    Hence the code $C_{D^c}$ is a Griesmer code. From Lemma \ref{minimal}, we have $C_{D^c}$ is a minimal code if
    \begin{equation*}
        \begin{aligned}
            \frac{\wt_{min}}{\wt_{max}}&=\frac{(2^n-1)\times(2^{n(m-1)}-2^{\sum_{i=0}^{n-1}|L_i|-n})}{(2^n-1)\times 2^{n(m-1)}}\\
            &=1-2^{\sum_{i=0}^{n-1}|L_i|-nm}\\
            &>\frac{2^n-1}{2^n},
        \end{aligned}
    \end{equation*}
    which is equivalent to $\sum_{i=0}^{n-1}|L_i|\leq nm-(n+1)$.
\end{proof}
\begin{remark}
    {\rm The cases $n=1,2,3$ in Theorem \ref{th3.3} are Example $3.4$ in \cite{hyun2020infinite}, Theorem $4.4$ in \cite{wu2022quaternary} and Theorem $3.9$ in
\cite{sagar2022linear} respectively}.
\end{remark}
\begin{example}
    {\rm Set $n=2,m=3$ and $L_0=\{1,2\},L_1=\{2,3\}\subseteq [3]=\{1,2,3\}$. Let $D=\Delta_{L_0}+\omega\Delta_{L_1}\subseteq \mathbb{F}_4^3$ so that $D^c=(\Delta_{L_0}^c+
    \omega\mathbb{F}_2^3)\bigsqcup(\Delta_{L_0}+\omega\Delta_{L_1}^c)$, where
    \begin{equation*}
        \begin{aligned}
            \Delta_{L_0}^c&=\{(0,0,1),(1,0,1),(0,1,1),(1,1,1)\},\\
            \Delta_{L_1}^c&=\{(1,0,0),(1,1,0),(1,0,1),(1,1,1)\}.
        \end{aligned}
    \end{equation*}
    Then the code $C_{D^c}$ is a $[48,4,36]$ linear code over $\mathbb{F}_4$ and it is a $3$-weight linear code and distance optimal by \cite{codetables}. In
    particular, $C_{D^c}$ has codewords of weights $0,36,40,48$}.
\end{example}
If all the $n$ subsets are nonempty and equal then we have the following result.
\begin{theorem}\label{3.4}
    Let $L$ be an nonempty subset of $[m]$ and let $D=\Delta_{L}+\omega\Delta_{L}+\cdots+\omega^{n-1}\Delta_{L}\subseteq \mathbb{F}_{2^n}^m$, where $m\in\mathbb{N}$.
    Then the code $C_{D^c}$ is a $2$-weight linear code over $\mathbb{F}_{2^n}$ of length $2^{nm}-2^{n|L|}$, dimension $m$ and distance
    $(2^n-1)\times(2^{n(m-1)}-2^{n(|L|-1)})$. In particular, $C_{D^c}$ have nonzero codewords of weights $(2^n-1)\times(2^{n(m-1)}-2^{n(|L|-1)})$ and
    $(2^n-1)\times 2^{n(m-1)}$. Furthermore, $C_{D^c}$ is a Griesmer code hence it is distance optimal. In fact, it is a minimal code if
    $n(m-|L|)\geq n+1$.
\end{theorem}
\begin{proof}
    It is a straightforward result from Theorem \ref{th3.2} and Theorem \ref{th3.3}.
\end{proof}
\begin{remark}\rm{
    The case $n=3$ in Theorem $\ref{3.4}$ is Theorem 3.10 in \cite{sagar2022linear}.}
\end{remark}
\begin{example}
    {\rm Set $n=2,m=4$ and $L=\{1,2\}\subseteq[4]=\{1,2,3,4\}$. Let $D=\Delta_{L}+\omega\Delta_{L}\subseteq \mathbb{F}_4^4$ so that $D^c=(\Delta_{L}^c+
    \omega\mathbb{F}_2^4)\bigsqcup(\Delta_{L}+\omega\Delta_{L}^c)$, where
    \begin{equation*}
        \begin{aligned}
            \Delta_{L}^c=&\{(0,0,1,0),(0,0,0,1),(1,0,1,0),(1,0,0,1),(0,1,1,0),(0,1,0,1),\\
            &(0,0,1,1),(1,1,1,0),(1,1,0,1),(1,0,1,1),(0,1,1,1),(1,1,1,1)\}.
        \end{aligned}
    \end{equation*}
    Then the code $C_{D^c}$ is a $[240,4,180]$ linear code over $\mathbb{F}_4$ and it is a $3$-weight linear code and distance optimal by \cite{codetables}. In
    particular, $C_{D^c}$ has codewords of weights $0,180,192$. Since $\frac{180}{192}>\frac{3}{4}$, it is a minimal code by Lemma $\ref{minimal}$.}
\end{example}
\section{The subfield code with respect to $C_{D^*}$ and $C_{D^c}$}
In this section, we study subfield codes with respect to $C_{D^*}$ and $C_{D^c}$ discussed in the previous section. In \cite{sagar2022linear}, the authors
gave the length and distance of the subfield code with respect to $C_{D^*}$ over $\mathbb{F}_{2^n}$. Now we also give the dimension of $C_{D^*}$.
\begin{theorem}\label{CD2}
    Let $L_i$ be nonempty subsets of $[m]$ and let $D=\Delta_{L_0}+\omega\Delta_{L_1}+\cdots+\omega^{n-1}\Delta_{L_{n-1}}\subseteq \mathbb{F}_{2^n}^m$. Then the
    subfield code $C_{D^*}^{(2)}$ with respect to $C_{D^*}$ is a 1-weight linear code over $\mathbb{F}_{2^n}$ of length $2^{\sum_{i=0}^{n-1}|L_i|}-1$, dimension
    $\sum_{i=0}^{n-1}|L_i|$ and distance $2^{\sum_{i=0}^{n-1}|L_i|-1}$. In particular, $C_{D^*}^{(2)}$ is a minimal code. Furthermore, it is a Griesmer code hence it
    is distance optimal.
\end{theorem}
\begin{proof}
    From Theorem \ref{subfield}, $C_{D^*}^{(2)}$ has defining set
    \begin{equation*}
        D^{(2)}=\{(d_0,d_1,\cdots,d_{n-1}):d_i\in \Delta_{L_i},0\leq i\leq n-1\}\subseteq (\mathbb{F}_2^m)^n.
    \end{equation*}
    Note that the length of $C_{D^*}^{(2)}$ is $|D^{(2)^*}|=|D^*|=2^{\sum_{i=0}^{n-1}|L_i|}-1$. Observe that the map
    $c_{D^*}^{(2)}:(\mathbb{F}_2^m)^n\longrightarrow C_{D^*}^{(2)}$ defined by $c_{D^*}^{(2)}(z)=(z\cdot d)_{d\in D^{(2)^*}}$ is a surjective linear transformation.
    Write $z=(\alpha_0,\alpha_1,\cdots,\alpha_{n-1})$ where $\alpha_i\in\mathbb{F}_2^m$, now we have
    \begin{equation}\label{51}
        \begin{aligned}
            \wt(c_{D^*}^{(2)}(z))
            &=|D|-\frac{1}{2}\sum_{d_0\in \Delta_{L_0}}\sum_{d_1\in \Delta_{L_1}}\cdots\sum_{d_{n-1}\in \Delta_{L_{n-1}}}\left(1+(-1)^{\sum_{i=0}^{n-1}
            \alpha_i\cdot d_i}\right)\\
            &=\frac{1}{2}|D|-\frac{1}{2}\prod_{i=0}^{n-1}\sum_{d_i\in \Delta_{L_i}}(-1)^{\alpha_i\cdot d_i}\\
            &=\frac{1}{2}|D|-\frac{1}{2}\chi_{\alpha_0}(\Delta_{L_0})\chi_{\alpha_1}(\Delta_{L_1})\cdots\chi_{\alpha_{n-1}}(\Delta_{L_{n-1}})\\
            &=2^{\sum_{i=1}^{n-1}|L_i|-1}(1-\varphi(\alpha_0|L_0)\varphi(\alpha_1|L_1)\cdots\varphi(\alpha_{n-1}|L_{n-1})).
        \end{aligned}
    \end{equation}
    Thus $C_{D^*}^{(2)}$ is a 1-weight linear code over $\mathbb{F}_{2^n}$ of distance $2^{\sum_{i=0}^{n-1}|L_i|-1}$. Note that
    \begin{equation*}
        \begin{aligned}
            \wt(c_{D^*}^{(2)}(z))=0
            &\Longleftrightarrow \varphi(\alpha_0|L_0)=\varphi(\alpha_1|L_1)=\cdots=\varphi(\alpha_{n-1}|L_{n-1})=1\\
            &\Longleftrightarrow \operatorname{Supp}(\alpha_i)\cap L_i=\emptyset,~ \forall~ i=0,1,\cdots,n-1.
        \end{aligned}
    \end{equation*}
    Hence $|\ker(c_{D^*}^{(2)})|=\prod_{i=0}^{n-1}2^{m-|L_i|}=2^{nm-\sum_{i=0}^{n-1}|L_i|}$.
    By the first isomorphism theorem of groups, we have $|C_{D^*}^{(2)}|=\frac{\left|(\mathbb{F}_2^m)^n\right|}{\left|\ker(c_{D^*}^{(2)})\right|}=2^{\sum_{i=0}^{n-1}|L_i|}$. Hence
    $\dim C_{D^*}^{(2)}=\sum_{i=0}^{n-1}|L_i|$. Lastly we have
    \begin{equation*}
        \begin{aligned}
            &\quad\sum_{i=0}^{\sum_{i=0}^{n-1}|L_i|-1}\left\lceil \frac{2^{\sum_{i=0}^{n-1}|L_i|-1}}{2^i} \right\rceil\\
            &=\sum_{i=0}^{\sum_{i=0}^{n-1}|L_i|-1}2^i\\
            &=2^{\sum_{i=0}^{n-1}|L_i|}-1.
        \end{aligned}
    \end{equation*}
    Therefore, $C_{D^*}^{(2)}$ is a Griesmer code.
\end{proof}
\begin{remark}
    {\rm The cases $n=2,3$ in theorem $\ref{CD2}$ are Proposition $5.1$ in \cite{wu2022quaternary} and Theorem $4.1$ in \cite{sagar2022linear} respectively.}
\end{remark}
\begin{example}{\rm
    Set $n=3,m=4$ and $L_0=\{1,2\},L_1=\{2,3\},L_2=\{3,4\}\subseteq [4]=\{1,2,3,4\}$. Consider the definition set $\Delta_{L_0}+\omega\Delta_{L_1}+\omega^2\Delta_{L_2}
    \subseteq \mathbb{F}_8^4$ so that $D^{(2)}=\{d_0,d_1,d_2:d_i\in D_i\}$ where
    \begin{equation*}
        \begin{aligned}
            \Delta_{L_0}&=\{(0,0,0,0),(1,0,0,0),(0,1,0,0),(1,1,0,0)\},\\
            \Delta_{L_1}&=\{(0,0,0,0),(0,1,0,0),(0,0,1,0),(0,1,1,0)\},\\
            \Delta_{L_2}&=\{(0,0,0,0),(0,0,1,0),(0,0,0,1),(0,0,1,1)\}.
        \end{aligned}
    \end{equation*}
    Then the code $C_{D^*}^{(2)}$ is a $[15,4,8]$ linear $1$-weight code over $\mathbb{F}_2$ which is distance optimal by \cite{codetables} and minimal.}
\end{example}
Suppose $D=\Delta_{L_0}+\omega\Delta_{L_1}+\cdots+\omega^{n-1}\Delta_{L_{n-1}}$, then from Equation \eqref{eq2} we have
\begin{equation*}
    \begin{aligned}
        D^{c}=&\left(\Delta_{L_0}^{c}+\omega \mathbb{F}_{2}^{m}+\cdots+\omega^{n-1} \mathbb{F}_{2}^{m}\right) \bigsqcup \\
        &\left(\Delta_{L_0}+\omega \Delta_{L_1}^{c}+\omega^{2} \mathbb{F}_{2}^{m} \cdots+\omega^{n-1} \mathbb{F}_{2}^{m}\right) \bigsqcup \\
        & \vdots \\
        &\left(\Delta_{L_0}+\omega \Delta_{L_1}+\cdots+\omega^{n-2} \Delta_{L_{n-2}}^{c}+\omega^{n-1} \mathbb{F}_{2}^{m}\right) \bigsqcup \\
        &\left(\Delta_{L_0}+\omega \Delta_{L_1}+\cdots+\omega^{n-2} \Delta_{L_{n-2}}+\omega^{n-1} \Delta_{L_{n-1}}^{c}\right),
    \end{aligned}
\end{equation*}
where $\bigsqcup$ indicates disjoint union.

If $z=(\alpha_0,\alpha_1,\cdots,\alpha_{n-1})\in (\mathbb{F}_2^m)^n$, then we have
\begin{equation*}
    \begin{aligned}
        \wt(c_{D^c}^{(2)}(z))
        &=|D^c|-\frac{1}{2}\chi_{\alpha_0}(\Delta_{L_0}^c)\chi_{\alpha_1}(\mathbb{F}_{2}^{m})\cdots\chi_{\alpha_{n-1}}(\mathbb{F}_{2}^{m})\\
        &\quad-\frac{1}{2}\chi_{\alpha_0}(\Delta_{L_0})\chi_{\alpha_1}(\Delta_{L_1}^c)\chi_{\alpha_2}(\mathbb{F}_{2}^{m})\cdots\chi_{\alpha_{n-1}}
        (\mathbb{F}_{2}^{m})\\
        &\quad-\cdots\\
        &\quad-\frac{1}{2}\chi_{\alpha_0}(\Delta_{L_0})\cdots\chi_{\alpha_{n-3}}(\Delta_{L_{n-3}})\chi_{\alpha_{n-2}}(\Delta_{L_{n-2}}^c)
        \chi_{\alpha_{n-1}}(\mathbb{F}_{2}^{m})\\
        &\quad-\frac{1}{2}\chi_{\alpha_0}(\Delta_{L_0})\cdots\chi_{\alpha_{n-2}}(\Delta_{L_{n-2}})\chi_{\alpha_{n-1}}(\Delta_{L_{n-1}}^c)\\
        &=\frac{1}{2}(|D^c|-2^{nm}\delta_{0,z})+\frac{1}{2}\chi_{\alpha_0}(\Delta_{L_0})\chi_{\alpha_1}(\Delta_{L_1})\cdots\chi_{\alpha_{n-1}}(\Delta_{L_{n-1}})\\
        &=2^{nm-1}\times(1-\delta_{0,z})-\frac{1}{2}|D|+\frac{1}{2}\chi_{\alpha_0}(\Delta_{L_0})\chi_{\alpha_1}(\Delta_{L_1})\cdots\chi_{\alpha_{n-1}}
        (\Delta_{L_{n-1}}).
    \end{aligned}
\end{equation*}
From Equation \eqref{51}, we have
\begin{equation*}
    \wt(c_{D^*}^{(2)}(z))=\frac{1}{2}|D|-\frac{1}{2}\chi_{\alpha_0}(\Delta_{L_0})\chi_{\alpha_1}(\Delta_{L_1})\cdots\chi_{\alpha_{n-1}}(\Delta_{L_{n-1}}),
\end{equation*}
thus we have the following result.
\begin{lemma}
    Let the symbols be the same as above. For $z=(\alpha_0,\alpha_1,\cdots,\alpha_{n-1})\in (\mathbb{F}_2^m)^n$, we have
    \begin{equation}\label{sub}
        \wt(c_{D^*}^{(2)}(z))+\wt(c_{D^c}^{(2)}(z))=2^{nm-1}\times(1-\delta_{0,z}).
    \end{equation}
\end{lemma}

\begin{theorem}\label{CDC}
    Suppose $L_i$ be nonempty subsets of $[m]$ such that at least one subset is proper where $m\in\mathbb{N}$. Let
    $D=\Delta_{L_0}+\omega\Delta_{L_1}+\cdots+\omega^{n-1}\Delta_{L_{n-1}}\subseteq \mathbb{F}_{2^n}^m$. Then the code  $C_{D^c}^{(2)}$ is a binary 2-weight linear code
    of length $2^{nm}-2^{\sum_{i=0}^{n-1}|L_i|}$, dimension $nm$ and distance $2^{nm-1}-2^{\sum_{i=0}^{n-1}|L_i|-1}$. In particular, $C_{D^c}^{(2)}$ have  nonzero
    codewords of weights $2^{nm-1}$ and $2^{nm-1}-2^{\sum_{i=0}^{n-1}|L_i|-1}$. Moreover, it is a Griesmer code and hence distance optimal. Further, it is a minimal
    code if $2^{\sum_{i=0}^{n-1}|L_i|}\leq nm-2$.
\end{theorem}
\begin{proof}
    Clearly, the length of the code $C_{D^c}^{(2)}$ is $|D^c|=2^{nm}-|D|=2^{nm}-2^{\sum_{i=0}^{n-1}|L_i|}$. For nonzero $z\in (\mathbb{F}_2^m)^n$, from Equation
    \eqref{sub} and Theorem \ref{CD2} we have
    \begin{equation*}
        \begin{aligned}
            \wt(c_{D^c}^{(2)}(z))&=2^{nm-1}-\wt(c_{D^*}^{(2)}(z))\\
            &>2^{nm-1}-2^{\sum_{i=1}^{n-1}|L_i|-1}\\
            &>0
        \end{aligned}
    \end{equation*}
    since $L_i$ is a proper subset of $[m]$ for some $i$. Thus $\ker c_{D^c}^{(2)}=0$ and then $\dim C_{D^c}^{(2)}=nm$. Now we have
    \begin{equation*}
        \begin{aligned}
            \sum_{i=0}^{nm-1}\left\lceil\frac{2^{nm-1}-2^{\sum_{i=0}^{n-1}|L_i|-1}}{2^i}\right\rceil&=\sum_{i=0}^{nm-1}\frac{2^{nm-1}}{2^i}-
            \sum_{i=0}^{nm-1}\left\lfloor\frac{2^{\sum_{i=0}^{n-1}|L_i|-1}}{2^i}\right\rfloor\\
            &=\sum_{i=0}^{nm-1}2^i-\sum_{i=0}^{\sum_{i=0}^{n-1}|L_i|-1}2^i\\
            &=(2^{nm}-1)-(2^{\sum_{i=0}^{n-1}|L_i|}-1)\\
            &=2^{nm}-2^{\sum_{i=0}^{n-1}|L_i|}.
        \end{aligned}
    \end{equation*}
    Therefore $C_{D^c}^{(2)}$ is a Griesmer code. By Lemma \ref{minimal}, we have
    \begin{equation*}
        \begin{aligned}
            \frac{\wt_{min}}{\wt_{max}}&=\frac{2^{nm-1}-2^{\sum_{i=0}^{n-1}|L_i|-1}}{2^{nm-1}}\\
            &=1-2^{\sum_{i=0}^{n-1}|L_i|-nm},
        \end{aligned}
    \end{equation*}
    thus it is a minimal code if $1-2^{\sum_{i=0}^{n-1}|L_i|-nm}>\frac{1}{2}$, which is equivalent to $\sum_{i=0}^{n-1}|L_i|\leq nm-2$.
\end{proof}
\begin{remark}{\rm
    The case $n=3$ in Theorem \ref{CDC} is Theorem 4.7 in \cite{sagar2022linear}.}
\end{remark}
\begin{example}{\rm
    Set $n=2,m=3$ and $L_0=\{1,2\},L_1=\{2,3\}\subseteq [3]=\{1,2,3\}$. Let $D=\Delta_{L_0}+\omega\Delta_{L_1}\subseteq \mathbb{F}_4^3$ so that $D^c=(\Delta_{L_0}^c+
    \omega\mathbb{F}_2^3)\bigsqcup(\Delta_{L_0}+\omega\Delta_{L_1}^c)$, where
    \begin{equation*}
        \begin{aligned}
            \Delta_{L_0}^c&=\{(0,0,1),(1,0,1),(0,1,1),(1,1,1)\},\\
            \Delta_{L_1}^c&=\{(1,0,0),(1,1,0),(1,0,1),(1,1,1)\}.
        \end{aligned}
    \end{equation*}
    Then the code $C_{D^c}^{(2)}$ is a $[48,6,24]$ linear $2$-weight code over $\mathbb{F}_2$ which is distance optimal by \cite{codetables}. In particular,
    $C_{D^c}^{(2)}$ has codewords of weights $0,24,32$. Since $\frac{24}{32}>\frac{1}{2}$, it is a minimal code by Lemma $\ref{minimal}$.}
\end{example}
\section{Conclusion}
In this article, we discussed the linear codes over $\mathbb{F}_{2^n}$ constructed from simplical complexes with one maximal element. This work is mainly inspired by
\cite{sagar2022linear} which studied the linear codes over $\mathbb{F}_{2^3}$ and also proposed some conjectures about the results over $\mathbb{F}_{2^n}$. In the
caculations of the weights of the codewords in these codes, we used the tools of LFSR sequences and boolean functions. We obtained five infinite families of distance
optimal codes and gave sufficient conditions for these codes to be minimal. These results generalized the cases in some previous researches which we listed in the
remarks of this article.

It is interesting to study the linear codes from simplical complexes with more than one maximal elements and construct more general optimal codes.
\vskip 4mm
\noindent {\bf Acknowledgement.} This work was supported by NSFC (Grant Nos. 12271199, 11871025).

\bibliographystyle{plain}

\end{document}